\documentclass[10pt,twocolumn]{article}
\usepackage{cite}
\usepackage{amsmath,amssymb,amsfonts}
\usepackage{algorithmic}
\usepackage{graphicx}
\usepackage{algorithm,algorithmic}
\usepackage{hyperref}
\hypersetup{hidelinks=true}
\usepackage{textcomp}
\usepackage{geometry}
\usepackage{enumerate}
\usepackage{overpic}
\usepackage{booktabs}
\usepackage{multirow}
\usepackage{makecell}
\usepackage{subfigure}

\newtheorem{assumption}{\bf Assumption}
\newtheorem{remark}{\bf Remark}
\newtheorem{definition}{\bf Definition}
\newtheorem{lemma}{\bf Lemma}
\newtheorem{corollary}{\bf Corollary}
\newtheorem{theorem}{\bf Theorem}

\begin{document}

\title{Fully Actuated Manifold Constraint Based Output Feedback Control for Input-Constrained Uncertain Nonlinear Systems}

\author{
	Dianrui Mu, Changchun Hua, Yafeng Li, Jiannan Chen, and Rao Wei
	\thanks{Authors Dianrui Mu, Changchun Hua, Yafeng Li, Jiannan Chen, and Rao Wei are with the Institute of Electrical Engineering, Yanshan University, Qinhuangdao, 066004, China (e-mail: mdr@stumail.ysu.edu.cn; cch@ysu.edu.cn; y.f.li@foxmail.com; cjn@ysu.edu.cn).}
}

\maketitle

\begin{abstract}
	This paper presents a low-complexity, model-free, output-feedback controller for a class of unknown time-varying nonlinear systems with unknown input constraints. The controller achieves the preset control accuracy when the actuator is not saturated and maintains flexible control accuracy after actuator saturation. This result extends existing constraint control methods for linear manifolds to a more general form, including the construction of nonlinear manifolds and various types of constraints, thereby achieving preset control accuracy within finite or fixed time. Additionally, flexible control under unknown saturation is achieved through the construction of an error-driven flexible constraint. Finally, second-order and higher-order control examples and simulations are provided.
\end{abstract}


\section{Introduction}
\label{sec:introduction}

\subsection{Motivation}
\noindent\textbf{B}ased on observations of different control methods below, control modeling and state acquisition, handling uncertainties, feedback form and performance guarantee are almost the most basic three issues of system control.

\subsubsection{Control modeling and state acquisition}
Control system modeling is a prerequisite for analyzing and designing. There are three modeling methods according to the state selection and the differential equation form. $\mathcal{M}odel$1: First-order state space equations \cite{swaroop2000dynamic} is a common modeling approach according to the physical system mechanisms. Nonlinear and uncertain terms are dispersed into first-order differential equations one by one, so they need to be addressed one by one. The acquisition of these physical states depends on sensors or state observers. Based on universal approximator, the output feedback problem of uncertain nonlinear systems can be solved in \cite{zhang2017fuzzy} and \cite{zhou2021event}, however, with a complex structure. The output feedback design without approximator usually has various limitations on the system: unit gain coefficient \cite{liu2022reduced}, known control gain \cite{li2023adaptive}, known nonlinear growth form \cite{zhang2019k}, known nonlinear boundary form \cite{wang2020global}, etc. $\mathcal{M}odel$2: Under certain conditions \cite{kanellakopoulos1991systematic}, strict feedback systems can be transformed into high-order fully actuated systems concerning output through differential homeomorphism in high-order fully actuated system approaches (HFASA) \cite{duan2021high2}. The system is transformed into a standard form, and all gain coefficients and unmatched dynamics are combined into total gain and total disturbance, allowing for unified processing. Combining differentiators, output feedback can be achieved \cite{cai2023active}, but it needs the control gain bounds to be known. $\mathcal{M}odel$3: Furthermore, considering the tracking control problem, establishing a fully actuated error system model can further simplify the acquisition of the states,
the derivative informations of the reference signal are no longer required in \cite{dimanidis2020output}.

\subsubsection{Handling uncertainties}
Dealing with internal and external uncertainties in the system to ensure its stability and controllability is a basic requirement of the controller. The common methods to handle uncertainty are as follows. $\mathcal{H}andle$1: Robust control including Bang-bang control, Parameter robust, Sliding mode control (SMC) \cite{utkin1977variable}, etc. It forces the error to converge by multiplying the upper bound of the uncertainty by the sign of the variable, but this also leads to controller chatting. $\mathcal{H}andle$2: Adaptive control including Parameter adaptive \cite{slotine1987adaptive}, Neural adaptive \cite{zhou2021event}, Fuzzy adaptive \cite{zhang2017fuzzy}, etc. The adaptive law can only estimate constant values. Although the adaptive method combining neural networks and fuzzy systems has universal approximation ability, its structure will become very complex. $\mathcal{H}andle$3: Disturbance estimation including Extended state observer (ESO) \cite{han2009pid}, Disturbance observer (DO), Time delay control (TDC), etc. Although the ESO can estimate the total disturbance, it needs to use the total gain value (at least a rough estimate) \cite{han2009pid}. $\mathcal{H}andle$4: Dynamic gain including Low-complexity prescribed performance control (PPC) \cite{bechlioulis2014low}, etc. It can offset the impact of uncertainty by adjusting gains dynamically. Its structure is simple, but there is incompatibility between multiple constraints when dealing with high-order systems \cite{cao2023prescribed}.

\subsubsection{Feedback form and performance guarantee}
The feedback form can directly determine the results of the Lyapunov stability analysis, thereby determining the performance of the system. Overall, the feedback form is a function of the states (or estimated states) and time. According to whether feedback variables are aggregated during use, feedback forms can be divided into two categories. $\mathcal{F}orm$1: Gradual feedback form is used in Backstepping, Cascade PID, etc. $\mathcal{F}orm$2: Aggregated feedback form is used in SMC, HFASA, etc.
To improve performance, two fast feedback forms have been developed. $\mathcal{F}ast$1: The power function feedback form is the function of the states, including finite-time control \cite{haddad2022asymptotic,hua2021event,min2017practically}, fixed-time control \cite{mi2022fixed,hua2016finite,zhou2024practical}, predefined time control \cite{munoz2019predefined}, etc. $\mathcal{F}ast$2: The time constrained feedback form is the function that dependents on the states and time, including funnel control \cite{ilchmann2002tracking}, PPC \cite{bechlioulis2008robust}, low-complexity PPC \cite{bechlioulis2014low}, prescribed time control \cite{shakouri2021prescribed,hua2021adaptive,cao2022practical}, etc.

The combination of gradual feedback, aggregated feedback, and fast feedback can form various feedback forms, such as fixed time sliding mode \cite{ni2017fixed}, prescribed time sliding mode \cite{chen2022prescribed}, prescribed performance HFASA \cite{zhang2022practical}, etc.
In recent years, a control method combining aggregated feedback and low-complexity PPC has been studied in \cite{bechlioulis2013output,wei2018robust,dimanidis2020output,lv2023distributed,song2016adaptive,cao2020adaptive}. Due to the manifolds used in them are all linear manifolds satisfying Hurwitz, we refer to it as linear manifold constraint control (LMCC). Among them, output feedback control is achieved in \cite{bechlioulis2013output} and \cite{dimanidis2020output}, while adaptive and integral control are combined in fault-tolerant PID (AFPID) control in \cite{song2016adaptive} and \cite{cao2020adaptive}. The performance of these methods requires complex formulas to be expressed so that the steady-state accuracy cannot be preset directly.

\subsection{Contributions}
By integrating $\mathcal{M}odel$3, $\mathcal{H}andle$4, $\mathcal{F}orm$2, and fast feedback form, a novel control method, fully actuated manifold constraint control (FAMCC), is proposed.
The major contributions are summarized below:
\begin{enumerate}[1)]
	\item The linear manifold constraint control method \cite{bechlioulis2013output,wei2018robust,dimanidis2020output,lv2023distributed,song2016adaptive,cao2020adaptive} is extended to nonlinear manifold constraint. It has the following characteristics:
	\begin{itemize}
		\item \textbf{Model-free:} The designed controller does not use any model information, particularly, the final output feedback controller relies solely on the system's output error.
		
		\item \textbf{Prescribed performance:} The steady-state accuracy can be directly preset in the controller, and the convergence time can be given, which can not be achieved in LMCC \cite{bechlioulis2013output,wei2018robust,dimanidis2020output,lv2023distributed,song2016adaptive,cao2020adaptive}.
		
		\item \textbf{Flexibility and robustness:} An error-driven flexible constraint is proposed to achieve flexible control under input constraints. Compared to saturation-driven flexible constraints \cite{trakas2024adaptive,berger2024input}, the proposed method does not require any knowledge on the input constraints.
		
		\item \textbf{Low complexity:} The controller design employs a low-complexity control method, which can ultimately be simplified into a very simple form.
	\end{itemize}
	
	\item For second-order systems, the ability of the proposed FAMCC under different combination schemes to converge to preset steady-state accuracy within finite time and fixed time is presented. Moreover, the discontinuity problem \cite{su2023comments} in the variable exponent coefficient fixed-time control method \cite{su2023comments,moulay2021robust,moulay2022fixed} is solved.
	
	\item For high-order systems, the fixed time convergence ability of the FAMCC based on the nonlinear manifold is proved through recursive analysis. It is shown that LMCC \cite{bechlioulis2013output,wei2018robust,dimanidis2020output,lv2023distributed} are special cases of the proposed FAMCC with linear manifold.
\end{enumerate}

\section{Problem Formulation}

\subsection{Notation}
For notation convenience, $\mathcal{R}$, $\mathcal{R}^+$, $\mathcal{R}^{n}$, and $\mathcal{R}^{m\times n}$ denote the real space, the nonnegative real space, the real $n$-dimensional space, and the real $m\times n$-dimensional space, respectively.
$\Vert\bullet\Vert$ is the Euclidean vector norm.
$\lceil z\rfloor^p=|z|^psign(z)$.
The arguments of the functions will be omitted or simplified whenever no confusion can arise from the context, e.g., $z(t)$ can be denoted by $z$.

\subsection{Smooth transition function}
It is often necessary to construct smooth functions when designing controllers to facilitate differential analysis and ensure that the controller is smooth or continuous. According to diverse needs, we introduce the following definitions:
\begin{definition}
	\emph{\textbf{Single-ended smooth transition function $\mathcal{T}(z)$}}
	has the following properties:
	\begin{enumerate}
		\item[a] $\mathcal{T}(z)=0, \forall z\leq0$, and $\mathcal{T}(1)=1$;
		\item[b] The $n$-th order derivative of $\mathcal{T}(z)$ is continuous on $(-\infty,1]$ and $\dot{\mathcal{T}}(z)\geq0$.
	\end{enumerate}
	\label{DefinitionSDSTF}
\end{definition}

\begin{definition}
	\emph{\textbf{Double-ended smooth transition function $\mathcal{S}(z)$}}
	has the following properties:
	\begin{enumerate}
		\item[a] $\mathcal{S}(z)=0, \forall z\leq0$, and $\mathcal{S}(z)=1, \forall z\geq1$;
		\item[b] The $n$-th order derivative of $\mathcal{S}(z)$ is continuous on $\mathcal{R}$ and $\dot{\mathcal{S}}(z)\geq0$.
	\end{enumerate}
	\label{DefinitionDESTF}
\end{definition}

\begin{definition}
	\emph{\textbf{Interval smooth transition function $\mathcal{U}(z)$}}
	has the following properties:
	\begin{enumerate}
		\item[a] $\mathcal{U}(0)=0$, and $\mathcal{U}(z)=1, \forall |z|\geq1$;
		\item[b] The $n$-th order derivative of $\mathcal{U}(z)$ is continuous on $\mathcal{R}$;
		\item[c] $\dot{\mathcal{U}}(z)\leq0, \forall z\leq0$ and $\dot{\mathcal{U}}(z)\geq0, \forall z\geq0$.
	\end{enumerate}
	\label{DefinitionISTF}
\end{definition}

\begin{remark}
	According to the definition, there are many selection forms for the smooth transition function mentioned above. Below is a set of examples:
	\begin{equation}
	\mathcal{T}(z)=
	\left\{
	\begin{array}{ll}
	0,&
	z\leq 0\\
	e^{\frac{z-1}{z}},&
	0<z<1
	\end{array}
	\right.
	\label{eq_trnsT}
	\end{equation}
	
	\begin{equation}
	\mathcal{S}(z)=
	\left\{
	\begin{array}{ll}
	0,&
	z\leq 0\\
	\frac{1}{e^{\frac{1-2z}{z(1-z)}}+1},&
	0<z<1\\
	1,&
	z\geq 1\\
	\end{array}
	\right.
	\label{eq_trnsS}
	\end{equation}
	
	\begin{equation}
	\mathcal{U}(z)=\mathcal{S}(|z|).
	\label{eq_trnsU}
	\end{equation}
	Any order derivatives of these functions are continuous.
\end{remark}

Based on single-ended smooth transition function, a prescribed performance function can be defined as
\begin{equation}
\rho(\rho_0,\epsilon,T_\rho)=(\rho_0-\epsilon)\mathcal{T}\left(\frac{T_\rho-t}{T_\rho}\right)+\epsilon
\label{eq_rho}
\end{equation}
where $\rho_0>\epsilon>0$ and $T_\rho>0$ are all positive constants. $\rho_0$, $\epsilon$, and $T_\rho$ are the constraint on the initial state, prescribed accuracy, and settling time respectively.

\subsection{System}
The following strict-feedback uncertain nonlinear system is investigated in this article:
\begin{equation}
	\left\{
	\begin{array}{l}
		\dot{x}_{oi}=g_i(\bar{x}_{oi},t)x_{o(i+1)}+f_i(\bar{x}_{oi},t)\\
		\dot{x}_{on}=g_n(\bar{x}_{on},t)u(v)+f_n(\bar{x}_{on},t)\\
		y_o=x_{o1}
	\end{array}
	\right.
	\label{eq_SFS}
\end{equation}
where $i=1,2,\cdots,n-1$ and $\bar{x}_{oi}=[x_{o1},\dots,x_{oi}]^T\in\mathcal{R}^i$. $\bar{x}_{on}=[x_{o1},\dots,x_{on}]^T\in\mathcal{R}^n$, $y_o\in\mathcal{R}$, $v\in\mathcal{R}$, and $u\in\mathcal{R}$ are the state vector, the system output, the to-be-designed controller, and actual control input provided by the actuator, respectively. The gain coefficients $g_i(\bar{x}_{oi},t)\in\mathcal{R}^i\times \mathcal{R}^+\rightarrow\mathcal{R}$ and unmatched dynamics $f_i(\bar{x}_{oi},t)\in\mathcal{R}^i\times \mathcal{R}^+\rightarrow\mathcal{R}$ are unknown time-varying nonlinear functions that locally Lipschitz in their arguments.

\begin{assumption}
	$u(v)=v$ for all time.
	\label{Assumption u=v}
\end{assumption}

\begin{assumption}
	The $i$th-order ($1\leq i\leq n-1$) derivatives of the reference signal $y_d$ can be obtained.
	\label{Assumption dyd}
\end{assumption}

\begin{assumption}
	\cite{bechlioulis2013output,wei2018robust,dimanidis2020output} The $i$th-order ($1\leq i\leq n-1$) derivatives of the system output $y_o$ are available.
	\label{Assumption dy}
\end{assumption}

\begin{assumption}
	The reference signal $y_d$ is a continuously differentiable and bounded function of time with bounded derivatives up to order $n$.
	\label{Assumption yd}
\end{assumption}

\begin{assumption}
	The initial states of system $x_{oi}(0)$ with $1\leq i \leq n$ are bounded.
	\label{Assumption x0}
\end{assumption}

\begin{assumption}
	For $1\leq i \leq n$, $g_i$, $f_i$ and those $(n-i)$th-order derivatives are continuous and bounded for $\|\bar{x}_{oi}\|\in L_\infty$.
	There exists an unknown positive constant $\underline{G}$ such that $\prod \limits_{i=1}^n g_i>\underline{G}$.
	\label{Assumption gf}
\end{assumption}

\begin{assumption}
	There are unknown time-varying input constraints $\underline{u}(t)<0<\bar{u}(t)$ on $u(v)$, defined as:
	\begin{equation}
	u(v)=\left\{
	\begin{aligned}
	&\underline{u},&v<\underline{u}\\
	&v,&\underline{u}\leq v\leq\bar{u}\\
	&\bar{u},&v>\bar{u}.
	\end{aligned}
	\right.,
	\label{eq_sat}
	\end{equation}
	Under these constraints, the system remains input-to-state stable (ISS) \cite{trakas2024adaptive} or input-to-state practically stable (ISpS) \cite{fotiadis2023input}.
	\label{Assumption sat}
\end{assumption}

\begin{remark}
	Considering the various possibilities of the inherent characteristics of the system, the reference signal, and the initial states, not all input constraints can ensure the system remains stable. This is particularly true for self-excited divergent systems, such as \( \dot{x} = x + u(v) \) with \( \bar{u}=-\underline{u} = 1 \), where if the reference signal or initial state leads to \( |x| > 1 \), no controller can drive the system back to stability. Therefore, in general, it is assumed that the system under input constraints remains ISS or ISpS.
\end{remark}

\emph{\textbf{Objective:}} For the system \eqref{eq_SFS}, under different assumptions as Table \ref{table_Objective}, the control objective is to design controllers such that the tracking error $z_1=y_o-y_d$ converges to prescribed accuracy $\epsilon_z$ within finite/fixed time whenever actuation limitations allow, or flexible prescribed accuracy when the input constraints are reached.

\begin{table*}[h!]
	\begin{center}
		\caption{Control Objective.}
		\begin{tabular}{c|l|l|l}
			\hline\hline
			\textbf{Type} & \textbf{Feedback} & \makecell[l]{\textbf{Assumptions}} & \makecell[l]{\textbf{Objective}} \\
			\hline
			\textbf{O1} & state feedback & Assumptions \ref{Assumption u=v}-\ref{Assumption gf} & \makecell[l]{prescribed accuracy} \\
			\hline
			\textbf{O2} & state feedback & Assumptions \ref{Assumption dyd}-\ref{Assumption sat} & \makecell[l]{flexible prescribed\\accuracy} \\
			\hline
			\textbf{O3} & output feedback & Assumptions \ref{Assumption yd}-\ref{Assumption sat} & \makecell[l]{flexible prescribed\\accuracy} \\
			\hline
			\bottomrule
		\end{tabular}
		\label{table_Objective}
	\end{center}
\end{table*}

\section{Algorithm Architecture}
\label{Algorithm_Architecture}
The design framework of this method includes four parts/steps:
\begin{enumerate}[Step A.]
	\item Fully actuated error system transformation.
	\item Fully actuated manifold design.
	\item Manifold constraint control.
	\item Manifold constraint control under input constraints.
	\item Differentiator-based manifold constraint control.
\end{enumerate}

\subsection{Fully Actuated Error System Transformation}
\label{subs_SF_A}
\noindent{\emph{\textbf{Step A1: Fully actuated system transformation}}}

After the differential homeomorphism transformation in the HFASA \cite{duan2021high2}, the system \eqref{eq_SFS} is transformed as follows:
\begin{equation}
	\left\{
	\begin{aligned}
		&\dot{x}_i=x_{i+1},\quad i=1,2,\cdots,n-1\\
		&\dot{x}_n=G_a(\bar{x}_{on},t)u+F_a(\bar{x}_{on},t)\\
		&y=x_1
	\end{aligned}
	\right.
	\label{eq_FAS}
\end{equation}
where total gain is $G_a=G_n$ and total disturbance is $F_a=F_n$ with $G_1=g_1$, $F_1=f_1$, $G_i=\prod \limits_{k=1}^i g_k$, and $F_i=\dot{F}_{i-1}+\dot{G}_{i-1}x_{oi}+G_{i-1}f_i$ when $i=2,\dots,n$.

\noindent{\emph{\textbf{Step A2: Fully actuated error system transformation}}}

Define system errors as:
\begin{equation}
	z_i=x_{i}-y_d^{(i-1)},\quad i=1,2,\cdots,n
	\label{eq_Se}
\end{equation}

The fully actuated errors system of \eqref{eq_FAS} is
\begin{equation}
	\left\{
	\begin{aligned}
		&\dot{z}_i=z_{i+1},\quad i=1,2,\cdots,n-1\\
		&\dot{z}_n=Gu+F
	\end{aligned}
	\right.
	\label{eq_FASe}
\end{equation}
where $G=G_a$ and $F=F_a-y_d^n$.

Define system the fully actuated state errors vector as:
\begin{equation}
\textbf{Z}(t)=\left[z_1(t),\dots,z_n(t)\right]^T\in\mathcal{R}^n.
\label{eq_Z}
\end{equation}

\subsection{Fully Actuated Manifold Design}
\label{subs_SF_B}
A manifold constructed by using all fully actuated errors is named fully actuated manifold.

\noindent{\emph{\textbf{Step B1: Iterative manifold design}}}

To construct a linear/nonlinear manifold with a general form for $n$th-order systems, negative feedback functions $h_{mi}(\bullet)$, where $i=1,\dots,n-1$, with the following properties, are introduced:
\begin{enumerate}
	\item $h_{mi}(0)=0$ and $\frac{\partial{h_{mi}(\bullet)}}{\partial{\bullet}}<0$, $\forall 1\leq i\leq n-1$.
	\item $h_{m(n-1)}(\bullet)$ is continuous and differentiable, and its derivative $\frac{\partial{h_{m(n-1)}(\bullet)}}{\partial{\bullet}}\in L_{\infty}$ for $\bullet\in L_{\infty}\bigcap \Omega_0^c$ with $\Omega_0^c:=\{\bullet||\bullet|> 0\}$.
	\item When $n>2$ and $i=1,\cdots,n-2$, $h_{mi}(\bullet)$ and its $(n-i)$th-order derivatives are continuous and bounded for bounded $\bullet$.
\end{enumerate}
$h_v$ is the inverse function of $h_{m(n-1)}$, which will be used later.

Referring to existing fast manifold design methods \cite{wu1998terminal,yang2020lyapunov,zhang2022nonsingular}, the fully actuated manifold of an $n$-order system can be constructed using iterative methods as follows
\begin{equation}
\left\{
	\begin{aligned}
	&s_1=z_1\\
	&s_i=\dot{s}_{i-1}-h_{m(i-1)}(s_{i-1}),\quad i=2,\cdots,n\\
	&s=s_n
	\end{aligned}
\right.
\label{eq_IS}
\end{equation}

It is equivalent to
\begin{equation}
s=z_n-\sum_{i=1}^{n-1}h_{mi}^{(n-1-i)}(s_i).
\label{eq_ISFM}
\end{equation}

\begin{remark}
	A general representation of the fully actuated manifold is a mapping $h_g\in \mathcal{R}^n\rightarrow\mathcal{R}$, formed as $s=h_g(\textbf{Z})$, which ensures that $\textbf{Z}$ tends towards $\textbf{0}$ when $s=0$.
	Referring to common manifold/filtered variable designs \cite{bechlioulis2013output,wei2018robust,dimanidis2020output,lv2023distributed,song2016adaptive,cao2020adaptive}, a linear fully actuated manifold is designed as $s=\sum_{i=1}^n a_iz_i$, where $a_i$ are chosen such that the polynomial $\sum_{i=1}^n a_i\partial^{i-1}$ is Hurwitz with $\partial$ as Laplacian operator.
\end{remark}

\noindent{\emph{\textbf{Step B2: Skewed manifold design}}}

First of all, introduce a positive exponent coefficients feedback curve as follows:
\begin{equation}
S_p(s_i):=\{(s_i,\dot{s}_i)|\dot{s_i}-h_{pi}(s_i)=0\}
\label{eq_Sp}
\end{equation}
where $h_{pi}(\bullet)$ is a function composed of one or more positive exponent coefficients feedback functions.

It is worth noting that simply constraining the fully actuated errors $\textbf{Z}$ near the linear manifold as \cite{bechlioulis2013output,wei2018robust,dimanidis2020output,lv2023distributed,song2016adaptive,cao2020adaptive} cannot directly preset steady-state accuracy by a control parameter. It also can not obtain the finite/fixed time convergence to the preset accuracy results by just replacing the linear manifold with a corresponding finite/fixed time manifold \cite{yang2020lyapunov,zhang2022nonsingular}.

Therefore, the following provides two skewed manifold designs to construct $h_{mi}$ based on positive exponent coefficients feedback $h_{pi}$.

\subsubsection{Smooth Skewed Manifold Design (SSMD)}

\begin{equation}
h_{mi}(s_i)=h_{pi}(s_i)+T_{1i}\epsilon_{si}
\label{eq_SSMD}
\end{equation}
\subsubsection{Nonsingular Skewed Manifold Design (NSMD)}

\begin{equation}
h_{mi}(s_i)=T_{2i}\left(h_{pi}(s_i)+T_{1i}\epsilon_{si}\right)-\left(1-T_{2i}\right)k_{pi}s_i
\label{eq_NSMD}
\end{equation}
where $\epsilon_{si}$ and $\epsilon_{zi}$ are positive constants which will be determined in specific use.
$T_{1i}=1-2\mathcal{S}\left(\frac{s_i+\epsilon_{zi}}{2\epsilon_{zi}}\right)$ is function constructed from double-ended smooth transition function defined in \eqref{eq_trnsS}.
$T_{2i}=\mathcal{U}\left(\frac{s_i}{\epsilon_{zi}}\right)$ is function constructed from interval smooth transition function defined in \eqref{eq_trnsU}.
$k_{pi}=k_{ppi}\frac{\epsilon_{si}-h_{pi}(\epsilon_{zi})}{\epsilon_{zi}}$ with positive constant $k_{ppi}$ that can adjust the slope of the $h_{mi}$ at $s_i=0$.

\begin{remark}
	Notably, as shown in Fig. \ref{fig_SMD}, if $\frac{\partial h_{pi}}{\partial s_i}$ is singular at $s_i=0$ with $\frac{\partial h_{pi}}{\partial s_i}(0)=-\infty$, SSMD can not change singularity, while NSMD can eliminate this singularity since $T_{2i}$ and its arbitrary order derivatives are $0$ at $s_i=0$. This guarantees that, when employing a nonlinear negative feedback function in conjunction with the iterative method \eqref{eq_IS} to construct nonlinear manifolds of order higher than two, singularities will not occur for all derivatives of $\mathcal{U}(z)$ at $z=0$ vanish.
\end{remark}

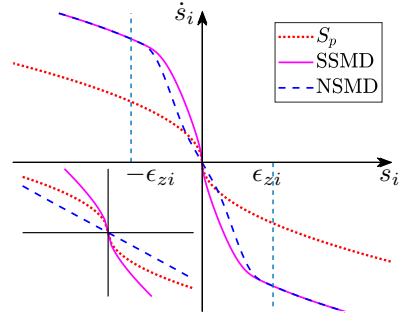
\begin{figure}[]
	\centering
	\begin{overpic}[width=0.8\columnwidth]{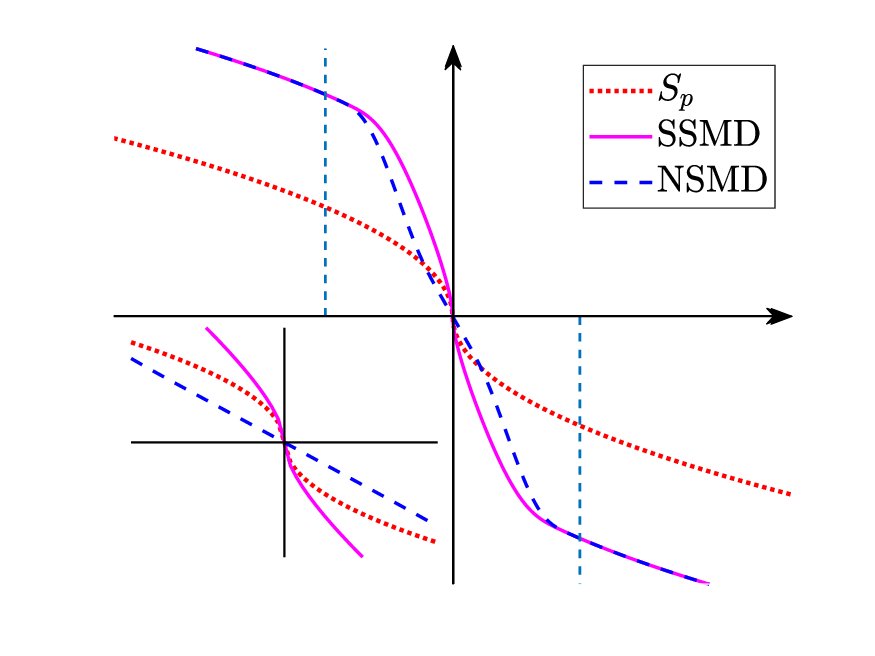}
		\put(88,35){$s_i$}
		\put(46,68){$\dot{s}_i$}
		\put(36,35){$-\epsilon_{zi}$}
		\put(62,35){$\epsilon_{zi}$}
	\end{overpic}
	\caption{Skewed manifold design diagram.}
	\label{fig_SMD}
\end{figure}

\subsection{Manifold Constraint Control}
\label{section_MCC}

The state with a manifold variable of $0$ is defined as zero manifold state 
\begin{equation}
\textbf{S}_0:=\{\textbf{Z}|s=0\}.
\label{eq_S0}
\end{equation}

\noindent{\emph{\textbf{Step C1: Define the boundaries of manifold constraint regions:}}}

Define the manifold translation function $s_m(x_m,y_m)$ as:
\begin{equation}
s_m(x_m,y_m)=\dot{s}_{n-1}-h_{m(n-1)}(s_{n-1}-x_m)-y_m
\label{eq_Sm}
\end{equation}

The upper boundary $\textbf{S}_U$ and lower boundary $\textbf{S}_L$ of the manifold constraint region are obtained by translating $\textbf{S}_0$ as:
\begin{equation}
\textbf{S}_U:=\{\textbf{Z}|s_m(x_U(t),y_U(t))=0\}
\label{eq_SU}
\end{equation}
\begin{equation}
\textbf{S}_L:=\{\textbf{Z}|s_m(x_L(t),y_L(t))=0\}
\label{eq_SL}
\end{equation}
in which $x_L(t)$, $x_U(t)$, $y_L(t)$, and $y_U(t)$ are time-varying boundary offsets that satisfy $x_L\leq 0\leq x_U$, $y_L\leq 0\leq y_U$, and $x_Ly_U=x_Uy_L$.

The manifold constraint region is:
\begin{equation}
\Psi_s:=\{\textbf{Z}|s_m(x_L,y_L)>0,s_m(x_U,y_U)<0\}.
\label{eq_Psis}
\end{equation}

To ensure that the initial state of the system is within the manifold constraint region $\Psi_s$ and establish constraints on performance, the boundary offsets should make the manifold constraint region large enough at the initial time, and the region $\Psi_s$ continues to shrink over time, and eventually reach its minimum at preset time $T_s$ and stop shrinking. Therefore, the prescribed performance function \eqref{eq_rho} is the ideal form of $-x_L(t)$, $x_U(t)$, $-y_L(t)$, and $y_U(t)$, e.g., as follows:
\begin{equation}
x_U=-x_L=\rho(k_0\rho_{x0},\epsilon_x,T_s)
\label{eq_xUxL}
\end{equation}
\begin{equation}
y_U=-y_L=\rho(k_0\rho_{y0},\epsilon_y,T_s)
\label{eq_yUyL}
\end{equation}
in which $\rho_{x0}=\max\{|s_{x0}|,\epsilon_x\}$ with $s_{x0}$ being the initial value of $|h_v(\dot{s}_{n-1})-s_{n-1}|$, and $\rho_{y0}=\max\{|s_{y0}|,\epsilon_y\}$ with $s_{y0}$ being the initial value of $|h_{m(n-1)}(s_{n-1})-\dot{s}_{n-1}|$. $\epsilon_x$ and $\epsilon_y$ are the final values of the lateral and vertical offset of the boundary, which can determine the size of the final constraint region $\Psi_s$. $k_0>1$ is a control coefficient that can adjust the initial control input since it can adjust the distance between the boundary $\textbf{S}_L,\textbf{S}_U$ and the initial state $(s_{n-1}(0),\dot{s}_{n-1}(0))$.

Manifold constraint can be divided into three types according to the direction of the boundary movement:

\subsubsection{Oblique Manifold Constraint (OMC)}
The boundary formed by $\textbf{S}_0$ oblique displacement: $x_Ly_U=x_Uy_L\neq0$
as in equation \eqref{eq_xUxL} and \eqref{eq_yUyL}.

\subsubsection{Longitudinal Manifold Constraint (LoMC)}
The boundary formed by $\textbf{S}_0$ vertical displacement: $x_L=x_U=0,y_Ly_U\neq0$ as in equation \eqref{eq_yUyL}.

\subsubsection{Lateral Manifold Constraint (LaMC)}
The boundary formed by $\textbf{S}_0$ horizontal displacement: $y_L=y_U=0,x_Lx_U\neq0$ as in equation \eqref{eq_xUxL}.

\begin{remark}
	For linear manifolds, the constraint boundaries obtained by using different translation methods are equivalent. Unlike linear manifolds, the constraint regions obtained by different translation methods are different for nonlinear manifolds.
	Although no examples of adapting to OMC are found in this paper, this case is presented as the most general form of manifold constraints for future reference.
\end{remark}
\begin{remark}
	Unlike common prescribed performance functions that choose a fixed value or $+\infty$ as the initial value $\rho_0$, this article uses multiples of the initial state values to construct the initial value of the performance function. This ensures that the initial state satisfies the constraints, while the initial value of the controller can be adjusted according to the coefficients $k_0$.
\end{remark}

\noindent{\emph{\textbf{Step C2: Calculate manifold constraint variables:}}}

In this step, a manifold constraint variable $\xi$ will be defined which can be used to represent the deviation of the state from the zero manifold state $\textbf{S}_0$ within the manifold constraint region $\Psi_s$.

The projection of zero manifold state $\textbf{S}_0$, upper bound of the manifold constraint $\textbf{S}_U$, and lower bound of the manifold constraint $\textbf{S}_L$ on the two-dimensional phase plane $(s_{n-1},\dot{s}_{n-1})$ are as follows:
\begin{equation}
S_0:=\{(x,y)|y=h_{m(n-1)}(x)\}
\label{eq_S0P}
\end{equation}
\begin{equation}
S_U:=\{(x,y)|y=h_{m(n-1)}(x-x_U)+y_U\}
\label{eq_SUP}
\end{equation}
\begin{equation}
S_L:=\{(x,y)|y=h_{m(n-1)}(x-x_L)+y_L\}
\label{eq_SLP}
\end{equation}

Define gradient lines as
\begin{equation}
S_G:=\left\{(x,y)|
\left\{
\begin{aligned}
&x_Uy=y_U(x-s_{n-1})+x_U\dot{s}_{n-1},&\text{OMC}\\
&x=s_{n-1},&\text{LoMC}\\
&y=\dot{s}_{n-1},&\text{LaMC}
\end{aligned}
\right.
\right\}
\label{eq_gradient}
\end{equation}

The intersection point between the gradient line $S_G$ and the manifold curve $S_0$ is $(x_c,y_c)$, which can be obtained by solving the following equation
\begin{equation}
x_U\left(h_{m(n-1)}(x_c)-\dot{s}_{n-1}\right)=y_U(x_c-s_{n-1})
\label{eq_xyc}
\end{equation}
in OMC.

Define the manifold constraint variable $\xi(\textbf{Z})$ as
\begin{equation}
\xi(\textbf{Z})=
	\left\{
	\begin{aligned}
	&2\frac{s_{n-1}-x_c-x_L}{x_U-x_L}-1,&\text{OMC}\\
	&2\frac{s-y_L}{y_U-y_L}-1,&\text{LoMC}\\
	&2\frac{s_{n-1}-h_v(\dot{s}_{n-1})-x_L}{x_U-x_L}-1,&\text{LaMC}
	\end{aligned}
	\right.
\label{eq_xi}
\end{equation}
When the state approaches the lower boundary $S_L$, $\xi$ tends towards $-1$, and when the state approaches the upper boundary $S_U$, $\xi$ tends towards $1$. Therefore, as long as the manifold constraint variable $\xi$ is constrained within $(-1,1)$, the $s$ can be constrained within the constraint region $\Psi_s$.

\noindent{\emph{\textbf{Step C3: Design manifold constraint controller:}}}

A manifold constraint controller is designed as follows:
\begin{equation}
v(\textbf{Z})=-k_u\Gamma(\xi(\textbf{Z}))
\label{eq_uT}
\end{equation}
where $k_u$ is positive constant. $\Gamma(\bullet)$ is a monotonic odd function mapped from interval $(-1,1)$ to interval $(-\infty,\infty)$, e.g., $\frac{\xi}{(1+\xi)(1-\xi)}$ and $\ln(\frac{1+\xi}{1-\xi})$.

\subsection{Manifold Constraint Control Under Input Constraints}

Taking into account the unknown input constraints, the error-driven flexible boundaries of the manifold constraints are designed as
\begin{equation}
\begin{aligned}
&\widetilde{x}_U=-\widetilde{x}_L=(1-\mathcal{T}_s)x_U+\mathcal{T}_s\frac{d_{sc}+\rho_e}{d_U}x_U\\
&\widetilde{y}_U=-\widetilde{y}_L=(1-\mathcal{T}_s)y_U+\mathcal{T}_s\frac{d_{sc}+\rho_e}{d_U}y_U
\end{aligned}
\label{eq_WxyUxyL}
\end{equation}
where $\mathcal{T}_s=\mathcal{S}\left(\frac{d_{sc}-(d_U-\rho_e)}{\rho_e}\right)$, $d_U=\sqrt{x_U^2+y_U^2}$, and $d_{sc}=\sqrt{(s_{n-1}-x_c)^2+(\dot{s}_{n-1}-y_c)^2}$ with $\rho_e\ll\sqrt{\epsilon_x^2+\epsilon_y^2}$ being a small positive constant.

The associated flexible constraint variable is defined as
\begin{equation}
\widetilde{\xi}(\textbf{Z})=
\left\{
\begin{aligned}
&2\frac{s_{n-1}-x_c-\widetilde{x}_L}{\widetilde{x}_U-\widetilde{x}_L}-1,&\text{OMC}\\
&2\frac{s-\widetilde{y}_L}{\widetilde{y}_U-\widetilde{y}_L}-1,&\text{LoMC}\\
&2\frac{s_{n-1}-h_v(\dot{s}_{n-1})-\widetilde{x}_L}{\widetilde{x}_U-\widetilde{x}_L}-1,&\text{LaMC}
\end{aligned}
\right.,
\label{eq_Wxi}
\end{equation}
and the corresponding flexible manifold constraint controller is designed as
\begin{equation}
v(\textbf{Z})=-k_u\Gamma\left(\widetilde{\xi}(\textbf{Z})\right).
\label{eq_WuT}
\end{equation}

\subsection{Differentiator-based Manifold Constraint Control}

The system \eqref{eq_FASe} can be rewritten as
\begin{equation}
\left\{
\begin{aligned}
&\dot{\textbf{Z}}=\mathcal{A}\textbf{Z}+\mathcal{B}(Gu+F)\\
&z_1=\mathcal{C}\textbf{Z}
\end{aligned}
\right.
\label{eq_FASM}
\end{equation}
where $\mathcal{A}=\begin{bmatrix}\textbf{0}_{n-1}&\textbf{I}_{n-1}\\0&\textbf{0}_{n-1}^T\end{bmatrix}, \mathcal{B}=\begin{bmatrix}\textbf{0}_{n-1}^T&1\end{bmatrix}^T, \mathcal{C}=\begin{bmatrix}1&\textbf{0}_{n-1}^T\end{bmatrix}$ in which $\textbf{I}_{n-1}\in\mathcal{R}^{(n-1)\times(n-1)}$ and $\textbf{0}_{n-1}\in\mathcal{R}^{n-1}$ are identity matrix and zero column vector, respectively.

A high gain differentiator used to obtain an estimate $\hat{\textbf{Z}}$ of the actual fully actuated errors $\textbf{Z}$ is designed as
\begin{equation}
\left\{
\begin{aligned}
&\dot{\hat{\textbf{Z}}}=\mathcal{A}\hat{\textbf{Z}}+H_\mu(z_1-\mathcal{C}\hat{\textbf{Z}})\\
&\hat{z}_1(0)=z_1(0)
\end{aligned}
\right.
\label{eq_HGD}
\end{equation}
with $H_\mu=diag(\frac{a_1}{\mu},\frac{a_2}{\mu^2},\dots,\frac{a_n}{\mu^n})$ where $\mu$ is a sufficiently small positive constant and $a_i(i=1,\dots,n)$ are chosen such that the polynomial $\partial^n+\sum_{i=1}^n a_i\partial^{n-i}$ is Hurwitz.

Replacing all the actual fully actuated errors $\textbf{Z}$ in the controller \eqref{eq_WuT}
with the estimate $\hat{\textbf{Z}}$, the output feedback controller is designed as
\begin{equation}
v(\hat{\textbf{Z}})=-k_u\Gamma\left(\widetilde{\xi}(\hat{\textbf{Z}})\right)
\label{eq_uOF}
\end{equation}

\section{Fully Actuated Manifold Constraint Theorem}

Similar to SMC, FAMCC has two stages: approaching the zero manifold $S_0$ before $T_s$ and moving along the vicinity of the zero manifold $S_0$ after $T_s$.

In this section, it is only proven that the controller has the ability to drive and constrain the system full drive error $\textbf{Z}$ to the vicinity of the zero manifold $S_0$, that is, the stage of approaching the manifold. The final ability to converge the output error $z_1$ to the preset accuracy $\epsilon_z$ in a finite/fixed time, i.e., the stage of moving along the vicinity of the manifold, will be demonstrated in subsequent corollaries.

\subsection{Lemma}

\begin{lemma}
	Considering the manifold design in \eqref{eq_IS} and the manifold constraint variable $\xi(\textbf{Z})$ in \eqref{eq_xi} for the system \eqref{eq_SFS} under Assumptions \ref{Assumption yd}-\ref{Assumption gf}, if $|\xi(\textbf{Z})|<1, \forall t\geq0$, then it holds that all $s_i(t)$ and $z_i(t), i=1,\dots,n$ remain bounded.
	\label{lemma1}
\end{lemma}

\textbf{Proof:} See Appendix \ref{Proof_Lemma}.$\hfill\blacksquare$

\subsection{State Feedback}

\begin{theorem}
	For the system \eqref{eq_SFS}, under Assumptions \ref{Assumption u=v}-\ref{Assumption gf}, if the controller \eqref{eq_uT} is applied, then it holds that:
	\begin{enumerate}[(\romannumeral1)]
		\item all closed-loop signals are bounded for all $t\geq0$;
		\item $\textbf{Z}(t)\in\Psi_s$ in \eqref{eq_Psis}, and $|\xi(\textbf{Z})|<1$ for all $t\geq0$;
		\item there exist positive constants $\widetilde{u}$ such that $|u|\leq\widetilde{u},\forall t\geq0$.
	\end{enumerate}
	\label{theoremSF}
\end{theorem}

\textbf{Proof:} See Appendix \ref{Proof_theoremSF}.$\hfill\blacksquare$

\subsection{State Feedback Under Input Constraints}

\begin{theorem}
	For the system \eqref{eq_SFS}, under Assumptions \ref{Assumption dyd}-\ref{Assumption sat}, if the controller \eqref{eq_WuT} is applied with sufficiently small \(\rho_e\), then it holds that:
	\begin{enumerate}[(\romannumeral1)]
		\item When the actuator output capability is sufficient, the conclusion is same as the Theorem \ref{theoremSF} (\romannumeral1) and (\romannumeral2);
		\item When the actuator output capability is insufficient, it can be ensured that the fully actuated errors $\textbf{Z}(t)$ remains within a flexible constraint set $\widetilde{\Psi}_s$
		\begin{equation}
		\widetilde{\Psi}_s:=\{\textbf{Z}|s_m(\widetilde{x}_L,\widetilde{y}_L)>0,s_m(\widetilde{x}_U,\widetilde{y}_U)<0\}.
		\label{eq_WPsis}
		\end{equation}
		The flexible constraint boundaries \eqref{eq_WxyUxyL} will be flexibly expanded after the controller \eqref{eq_WuT} reaches input constraints \eqref{eq_sat}, ensuring the controller \eqref{eq_WuT} is well-defined and the closed-loop system continues to operate effectively. Furthermore, the constraint boundaries are fully restored to their original values (i.e., \(\widetilde{x}_U=-\widetilde{x}_L=x_U=-x_L\) and \(\widetilde{y}_U=-\widetilde{y}_L=y_U=-y_L\)) when the controller exits saturation.
	\end{enumerate}
	\label{theoremSFIC}
\end{theorem}

\textbf{Proof:} See Appendix \ref{Proof_theoremSFIC}.$\hfill\blacksquare$

\begin{remark}
	It is worth noting that the designed flexible controller has the following new features:
	\begin{enumerate}[a.]
		\item The flexible constraint designed does not rely on input constraint information as \cite{trakas2024adaptive,berger2024input}, enabling a more thorough model-free control approach;  
		\item Unlike existing adaptive flexible boundaries \cite{trakas2024adaptive,berger2024input}, which recover slowly and cannot return to the original constraints within finite time, the proposed flexible controller can immediately restore the original constraints once saturation ends.
		\item The required parameter \(\rho_e\) is sufficiently small and can be easily tuned for successful implementation. Even if the selected \(\rho_e\) is not sufficiently small, the controller will continue to increase until the input constraints are reached, thereby ensuring the full utilization of the actuator's output capacity.
	\end{enumerate}
\end{remark}

\subsection{Output Feedback Under Input Constraints}
\label{subsection_OutputFeedback}

\begin{theorem}
	For the system \eqref{eq_SFS}, under Assumptions \ref{Assumption yd}-\ref{Assumption sat}, if the differentiator \eqref{eq_HGD} and the controller \eqref{eq_uOF} are applied, then there exits a constant $\bar{\mu}\in(0,1)$ such that the conclusion is same as Theorem \ref{theoremSFIC} when $0<\mu<\bar{\mu}$.
	\label{theoremOF}
\end{theorem}

\textbf{Proof:} See Appendix \ref{Proof_theoremOF}.$\hfill\blacksquare$

\begin{remark}
	The proof of the \textbf{Theorem} \ref{theoremOF} can also refer to \cite{bechlioulis2013output} and \cite{dimanidis2020output}.
	Since $z_1=y_o-y_d$ is known, setting the initial value of the differentiator $\hat{z}_1(0)$ to be the same as $z_1(0)$ in \eqref{eq_HGD} can make the estimated value converge faster.
\end{remark}

\subsection{Corollaries}

According to the \textbf{Theorem} \ref{theoremSF}, \textbf{Theorem} \ref{theoremSFIC}, and \textbf{Theorem} \ref{theoremOF}, the controller \eqref{eq_uT}, \eqref{eq_WuT}, and \eqref{eq_uOF} only implements constraints on the manifold, i.e., $\textbf{Z}$ always stays within $\Psi_s$.
The performance of the system output is related to the constructed manifold, the selected manifold constraint type, and their parameters.
The controller design provided in the section \ref{Algorithm_Architecture}, including two skewed types (SSMD / NSMD) and three manifold constraint types (OMC / LoMC / LaMC), can form many controller combinations by combining different feedback functions.

For ease of understanding, we will first demonstrate the controller's ability to converge to preset accuracy $z_1<\epsilon_z$ in finite and fixed time for second-order systems in sections \ref{Second-order Systems}. Further, the fixed time convergence ability of the controller for high-order systems is demonstrated in sections \ref{subs_SS_RFC}. The corresponding corollaries will be provided and proven. The configuration used for the control of each section is shown in the TABLE \ref{table_parameters} below.

\begin{table*}[h!]
	\begin{center}
		\caption{Composition of controllers in section \ref{Second-order Systems}.}
		\begin{tabular}{c|l|l|l} 
			\hline\hline
			\textbf{Section} & \textbf{Feedback function} & \makecell[l]{\textbf{Constraint}\\\textbf{types}} & \makecell[l]{\textbf{Skewed}\\\textbf{types}} \\
			\hline
			\multicolumn{4}{c}{\ref{Second-order Systems} Second-order Systems}\\
			\hline
			\ref{subs_SS_Fn} & Finite-time & LaMC & none \\
			\hline
			\ref{subs_SS_Fxe} & \makecell[l]{Variable exponent \\ coefficients fixed-time} & LoMC & SSMD \\
			\hline
			
			\multicolumn{4}{c}{\ref{High-order Systems} High-order Systems}\\
			\hline
			\ref{subs_SS_LLFMC} & \makecell[l]{Linear fully actuated\\ manifold} & \makecell[l]{OMC/LoMC\\/LaMC} & none \\
			\hline
			\ref{subs_SS_RFC} & \makecell[l]{Iterative fixed-time\\ fully actuated manifold} & LoMC & NSMD \\
			\hline
			\bottomrule
		\end{tabular}
		\label{table_parameters}
	\end{center}
\end{table*}

\begin{remark}
	It is worth noting that a very direct idea
	is to directly impose a constraint on manifold $s$. This is equivalent to LoMC without skewed design, as explained in the subsequent simplification of \eqref{eq_us}. However, it is difficult to preset control accuracy through control parameters by only $s$-constraint. Existing research \cite{bechlioulis2013output} and \cite{song2016adaptive} has provided the relationship between accuracy and control parameters such that it must reverse deduce control parameters according to accuracy requirements when using only $s$-constraint. Furthermore, fast sliding mode (nonlinear manifold), such as \cite{ni2017fixed} and \cite{moulay2021robust}, is a mature method to achieve faster convergence speed. However, it is even more difficult to provide the relationship between performance and control parameters in only $s$-constraint, so skewed manifold design methods in \emph{\textbf{Step B2}} and manifold constraint methods in \emph{\textbf{Step C1}} are designed specially in this article.
	By combining the fast manifold and LaMC, it is easy to preset accuracy and provide the convergence time as shown in section \ref{subs_SS_Fn}. However, considering that solving the inverse function of complex functions can be very complex, this approach is only suitable for simple fast manifolds.
	To achieve faster convergence, more complex fast manifolds are often required, which is the starting point of designing skewed manifolds in this article. Many fast manifolds (or SMC), such as \cite{ni2017fixed} and \cite{moulay2021robust}, can preset accuracy and provide the setting time by combining with skewed design, that will be presented in the example given in section \ref{subs_SS_Fxe} and \ref{subs_SS_RFC}.
\end{remark}

\section{Second-order Systems Control}
\label{Second-order Systems}

To present our design ideas more clearly, the case of system \eqref{eq_SFS} with $n=2$ (a second-order system) is considered at the first.
The positive exponent coefficients feedback has many forms that conform to the definition of negative feedback function, such as function with low exponent coefficients in finite-time control\cite{hua2021event}, function composed of low and high exponent coefficients function in fixed-time feedback\cite{ni2017fixed} and variable exponent coefficients feedback\cite{moulay2021robust}.

\subsection{LaMC Based Finite-time PPC}
\label{subs_SS_Fn}

Finite-time feedback law is selected as negative feedback function $h_{m1}$:
\begin{equation}
h_{m1}(s_1)=-k_c\lceil s_1\rfloor^p,0<p<1
\label{eq_fn}
\end{equation}

Under LaMC, the controller \eqref{eq_uOF} can be simplified as follows:
\begin{equation}
\begin{aligned}
v&=-k_u\ln\left(\frac{1+\xi}{1-\xi}\right)\\
&=-k_u\ln\left(\frac{1+(2\frac{s_{n-1}-h_v(\dot{s}_{n-1})-\widetilde{x}_L}{\widetilde{x}_U-\widetilde{x}_L}-1)}{1-(2\frac{s_{n-1}-h_v(\dot{s}_{n-1})-\widetilde{x}_L}{\widetilde{x}_U-\widetilde{x}_L}-1)}\right)\\
&=-k_u\ln\left(\frac{s_{n-1}-h_v(\dot{s}_{n-1})-\widetilde{x}_L}{\widetilde{x}_U-s_{n-1}+h_v(\dot{s}_{n-1})}\right)
\end{aligned}
\label{eq_us1}
\end{equation}

It can be observed that $s_{n-1}-h_v(\dot{s}_{n-1})=s_{n-1}+\frac{1}{k_c}\lceil \dot{s}_{n-1}\rfloor^{\frac{1}{p}}$ is a common form of nonsingular finite-time sliding surface. Defined it as
\begin{equation}
s_{fn}=s_{n-1}+\frac{1}{k_c}\lceil \dot{s}_{n-1}\rfloor^{\frac{1}{p}},0<p<1
\label{eq_nfn}
\end{equation}
The controller \eqref{eq_us1} can be further simplified as
\begin{equation}
v=-k_u\ln\left(\frac{s_{fn}-\widetilde{x}_L}{\widetilde{x}_U-s_{fn}}\right),
\label{eq_ufn}
\end{equation}
while $\widetilde{x}_L$ and $\widetilde{x}_U$ can be simplified as
\begin{equation}
\begin{aligned}
&\widetilde{x}_U=-\widetilde{x}_L\\
=&(1-\mathcal{T}_x)x_U+\mathcal{T}_x(|h_v(\dot{s}_{n-1})-s_{n-1}|+\rho_e)\\
=&(1-\mathcal{T}_x)x_U+\mathcal{T}_x(|s_{fn}|+\rho_e)
\end{aligned}
\label{eq_WxUxL}
\end{equation}
with $\mathcal{T}_x=\mathcal{S}\left(\frac{|s_{fn}|-(x_U-\rho_e)}{\rho_e}\right)$.

This indicates that the LoMC on finite-time manifolds with \eqref{eq_fn} is equivalent to a constraint on nonsingular finite-time manifold value $s_{fn}$ in \eqref{eq_nfn}.

\begin{corollary}
	\label{corollary_finite}
	For the system \eqref{eq_SFS} with $n=2$ under Assumptions \ref{Assumption yd}-\ref{Assumption sat}, if the high gain differentiator \eqref{eq_HGD} and the controller \eqref{eq_ufn} (which is equivalent to \eqref{eq_uOF} composed of \eqref{eq_IS}, \eqref{eq_xUxL}, \eqref{eq_xi}, and \eqref{eq_fn}) is used with parameter selection $\epsilon_x=\epsilon_z$ in \eqref{eq_xUxL} and sufficiently small $\rho_e,\mu$, the close-loop system has the following properties:
	\begin{enumerate}[(\romannumeral1)]
		\item When the actuator output capability is sufficient, the setting time $T_{fn}$ for $z_1$ to converge to the accuracy $|z_1|<\epsilon_z$ satisfies the following inequality:
		\begin{equation}
			T_{fn}< T_s+T_{fn}'
			\label{eq_fnT}
		\end{equation}
		with $T_{fn}'=\frac{z_1^{1-p}(0)}{k_c(1-p)}$;
		\item When the actuator output capability is insufficient, the system achieves a flexible control accuracy $|z_1|<\widetilde{x}_U$
		within the same finite time $T_{fn}$, and after the controller exits saturation, it recovers to the preset accuracy $\epsilon_z$ within a finite time $\left(\max\{T_e,T_s\}+T_{fn}'\right)$ where $T_e$ is the time when the system exits saturation.
	\end{enumerate}
\end{corollary}

\begin{figure}[]
	\centering
	\begin{overpic}[width=0.8\columnwidth]{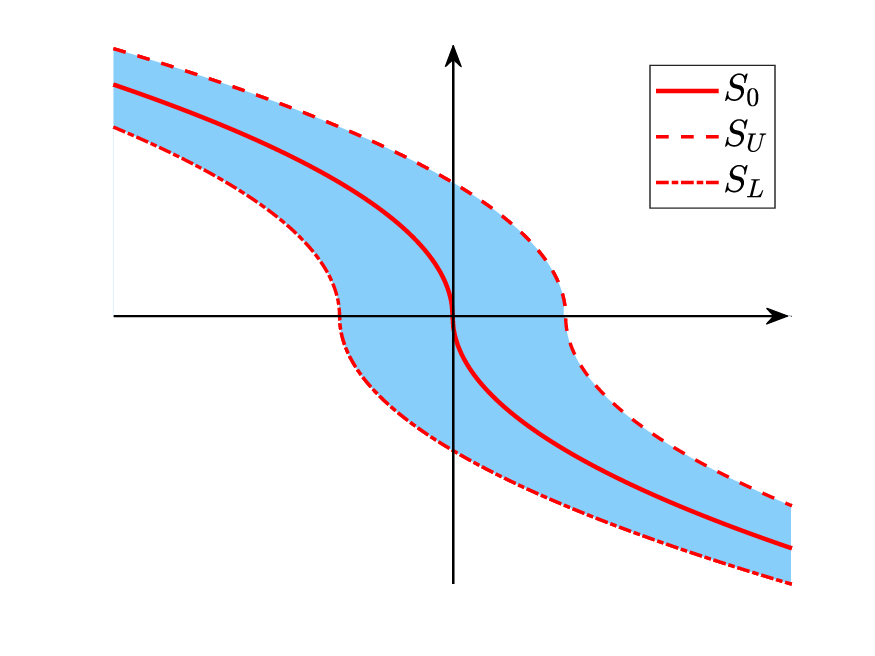}
		\put(88,35){$z_1$}
		\put(46,68){$z_2$}
		\put(36,35){$-\epsilon_x$}
		\put(62,35){$\epsilon_x$}
	\end{overpic}
	\caption{LaMC based finite-time manifold constraint diagram after $T_s$.}
	\label{fig_FiniteS}
\end{figure}
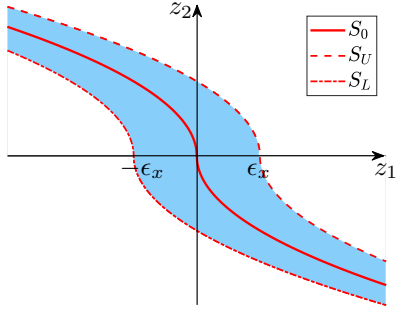

\textbf{Proof:}
According to \textbf{Theorem} \ref{theoremOF}, if the controller is used with $u=v,\forall t\geq0$, $\textbf{Z}$ always stays within $\Psi_s$, i.e., $\textbf{Z}$ will stays within the blue area in the Fig. \ref{fig_FiniteS} after $T_s$. Therefore, the following inequality holds:
\begin{equation}
\left\{
\begin{aligned}
&z_2>k_c|z_1+\epsilon_x|^p,&z_1<-\epsilon_x\\
&z_2<-k_c|z_1-\epsilon_x|^p,&z_1>\epsilon_x
\end{aligned}
\right.
\label{eq_fndz}
\end{equation}

Referring to the finite time control theorem \cite{hua2021event}, $z_1$ will converge to $z_1>-\epsilon_x$ in finite time for $z_1(0)<0$ and converge to $z_1<\epsilon_x$ in finite time for $z_1(0)>0$.
The convergence time $T_{fn}$ is expressed as \eqref{eq_fnT}, which is related to $T_s$, the initial state and the parameters of $h_{m1}$ in \eqref{eq_fn}. The final accuracy is determined by $\epsilon_x$ which can be set as $\epsilon_z$.

When the controller becomes saturated and the manifold constraint boundary is expanded, the system state will enter the flexible constraint set $\widetilde{\Psi}_s$ \eqref{eq_WPsis} as the design in \eqref{eq_WxyUxyL}. Similar to the above analysis, the flexible accuracy is $\widetilde{x}_U$. After exiting saturation, if the exit time $T_e$ is earlier than the preset time $T_s$, the time to converge to the preset accuracy $\epsilon_{z}$ remains $T_{fn}$. If the exit time is later than the preset time $T_e\geq T_s$, the time to restore the preset accuracy will not exceed \(\left(T_e+T_{fn}'\right)\).
$\hfill\blacksquare$

The following system is used in the simulation:
\begin{equation}
\left\{
\begin{aligned}
\dot{x}_1=&(1+0.1\cos(x_1)+0.1\sin(t))x_2-x_1+\sin(t)\\
\dot{x}_2=&(2+0.2\sin(x_1x_2)+0.1\cos(t))u(v)\\&-(x_1^2-1)x_2+\sin(t)+\cos(2t)\\
y=&x_1,\underline{u}=-100,\bar{u}=100,\\
x_1(&0)=-1,x_2(0)=0.5,y_d(t)=sin(t).\\
\end{aligned}
\right.
\label{eq_System_sim}
\end{equation}

\begin{figure}[]
	\centering
	\subfigure[]
	{
		\begin{overpic}[width=0.45\columnwidth]{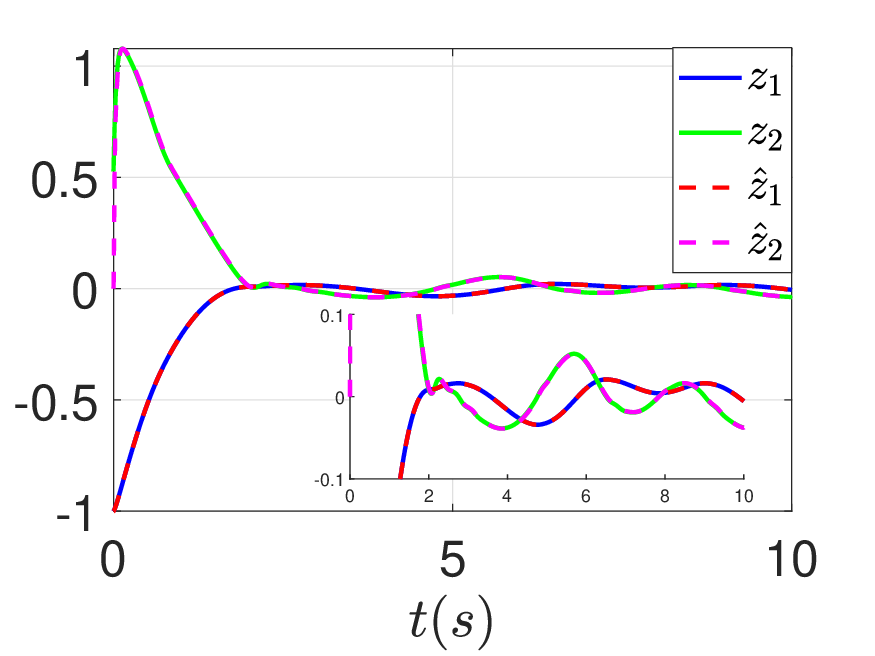}
		\end{overpic}
		\label{fig_Finitez}
	}\hfill
	\subfigure[]
	{
		\begin{overpic}[width=0.45\columnwidth]{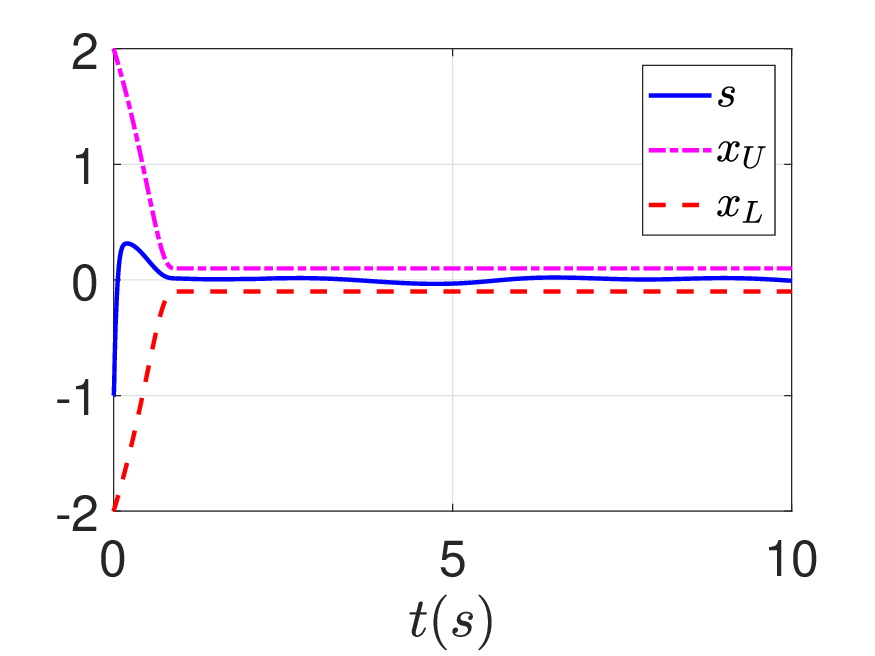}
		\end{overpic}
		\label{fig_Finites}
	}\\
	\subfigure[]
	{
		\begin{overpic}[width=0.45\columnwidth]{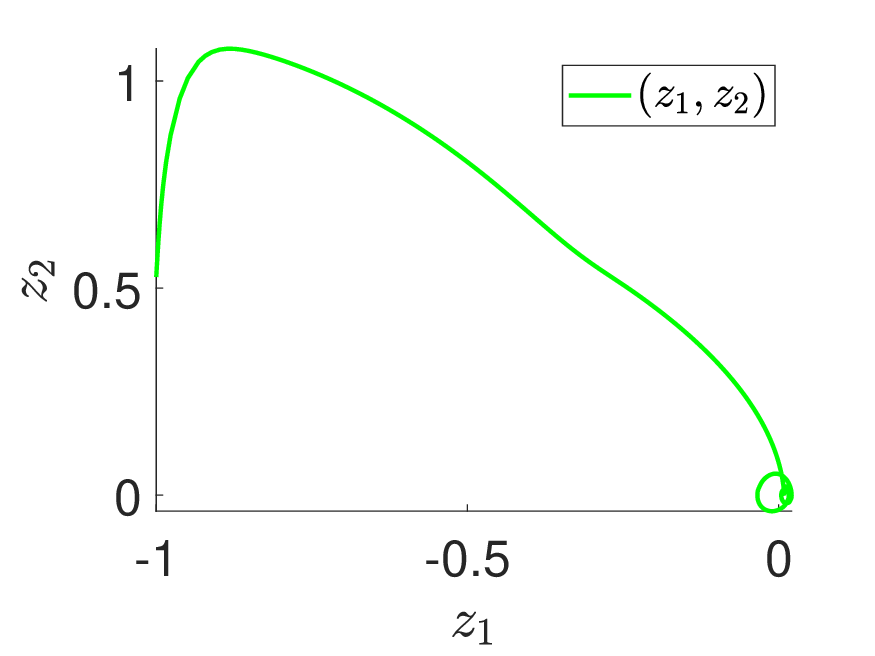}
		\end{overpic}
		\label{fig_Finitem}
	}\hfill
	\subfigure[]
	{
		\begin{overpic}[width=0.45\columnwidth]{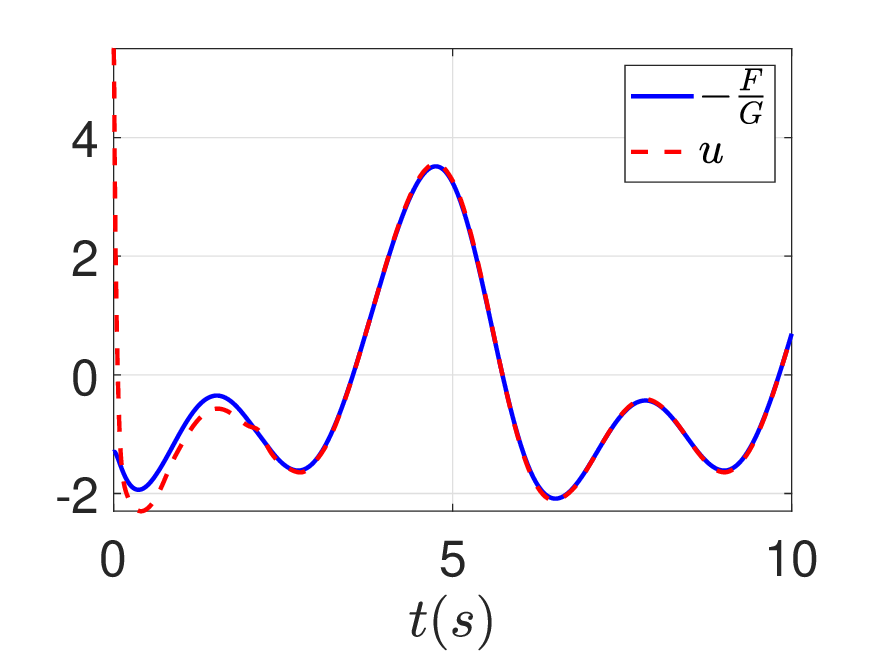}
		\end{overpic}
		\label{fig_Finiteu}
	}
	\caption{Simulation results of the finite-time PPC with LaMC. (a) The variables $z_1$, $z_2$, $\hat{z}_1$, and $\hat{z}_2$ versus time. (b) The variables $s_{fn}$, $x_U$, and $x_L$ versus time. (c) Phase trajectory $(z_1,z_2)$. (d) The evolution of $-\frac{F}{G}$ and $u$.}
	\label{fig_FiniteA}
\end{figure}

\begin{figure}[]
	\centering
	\subfigure[]
	{
		\begin{overpic}[width=0.45\columnwidth]{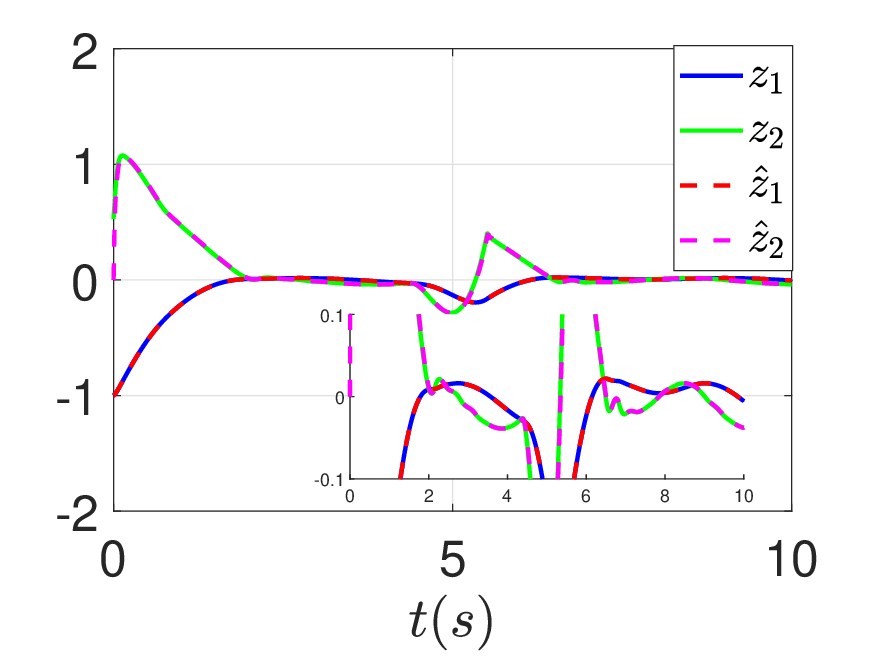}
		\end{overpic}
		\label{fig_FiniteSz}
	}\hfill
	\subfigure[]
	{
		\begin{overpic}[width=0.45\columnwidth]{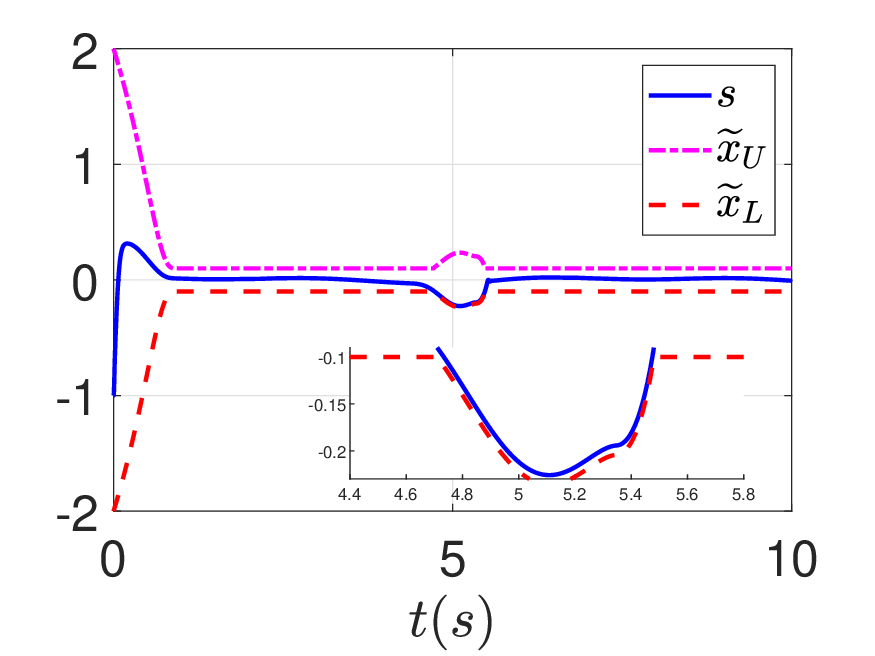}
		\end{overpic}
		\label{fig_FiniteSs}
	}\\
	\subfigure[]
	{
		\begin{overpic}[width=0.45\columnwidth]{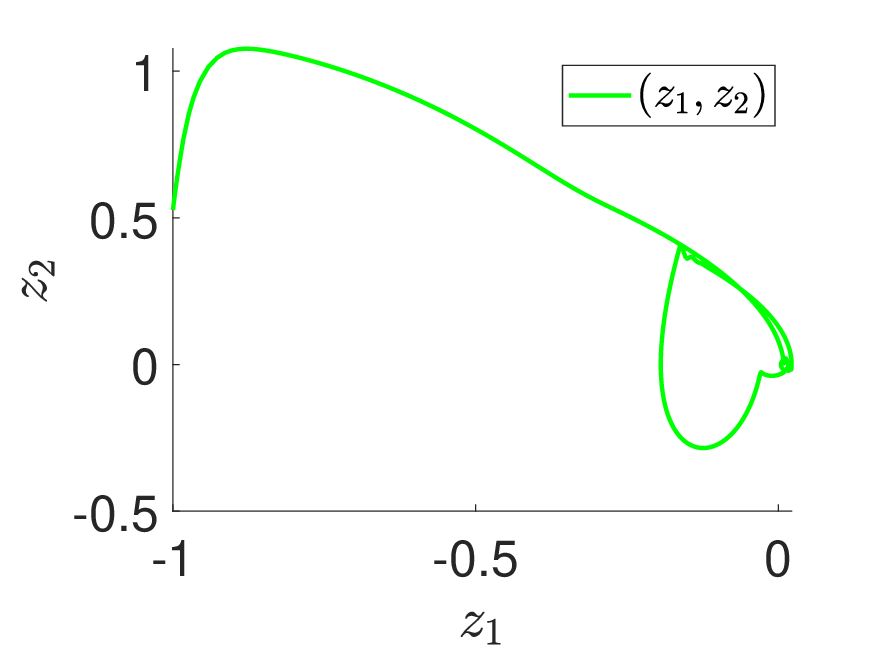}
		\end{overpic}
		\label{fig_FiniteSm}
	}\hfill
	\subfigure[]
	{
		\begin{overpic}[width=0.45\columnwidth]{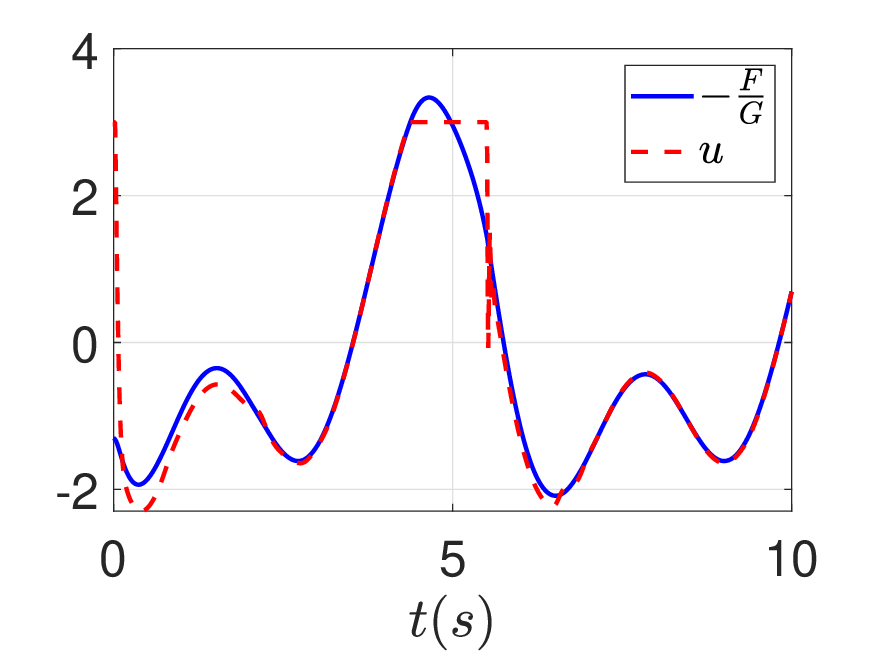}
		\end{overpic}
		\label{fig_FiniteSu}
	}
	\caption{Simulation results of the finite-time PPC with LaMC and saturated actuator. (a) The variables $z_1$, $z_2$, $\hat{z}_1$, and $\hat{z}_2$ versus time. (b) The variables $s_{fn}$, $\widetilde{x}_U$, and $\widetilde{x}_L$ versus time. (c) Phase trajectory $(z_1,z_2)$. (d) The evolution of $-\frac{F}{G}$ and $u$.}
	\label{fig_FiniteSatA}
\end{figure}

The controller parameters are selected as follows: $k_u=5,k_0=2,\epsilon_x=0.1,T_s=1,\rho_e=0.01$ in controller, $k_c=1,p=0.5$ in \eqref{eq_fn}, and $a_1=4,a_2=4,\mu=0.01$ in differentiator \eqref{eq_HGD}, therefore $T_{fn}<3$.

The simulation results are shown in 
Fig. \ref{fig_Finitez}-Fig. \ref{fig_Finiteu}.
It can be intuitively seen that $s_{fn}$ always evolves within the prescribed bounds $x_U$ and $x_L$, and $z_1$ achieves the prescribed accuracy of $0.1$ within the settling time of $1.9s<T_{fn}$.

With the control parameters set as described above, the input constraint in the system \eqref{eq_System_sim} is reduced to \( \underline{u}=-2.5,\bar{u}=3 \). The simulation results are shown in Fig \ref{fig_FiniteSatA}. The actuator output \( u \) remains within the input constraint at all times as Fig. \ref{fig_FiniteSu}. The flexible constraint $\widetilde{x}_U,\widetilde{x}_L$ expands as \( |s| \) increases, and it restores to its original preset value after exiting saturation as Fig. \ref{fig_FiniteSs}. The control accuracy is also able to recover to the preset accuracy within a finite time as Fig. \ref{fig_FiniteSz}.

\subsection{LoMC and SSMD Based Variable Exponent Coefficients Fixed-time PPC}
\label{subs_SS_Fxe}

A fixed time sliding surface with variable exponential coefficients is proposed in \cite{moulay2021robust} and \cite{moulay2022fixed}, but it is discontinuous \cite{su2023comments}. An improved form is proposed below:
\begin{subequations}
\begin{align}
&h_{p1}(s_1)=-k_c\lceil s_1\rfloor^p\\
&p=\frac{as_1^2}{1+bs_1^2}+r_b\\
&a=(r_t-r_b)b,b=\frac{r_1-r_b}{r_t-r_1}
\end{align}
\label{eq_fxe}
\end{subequations}
\begin{equation}
	z_2=-k_c\lceil z_1\rfloor^p+2\epsilon_y
\end{equation}
\begin{equation}
	z_2=-k_c\lceil z_1\rfloor^p
\end{equation}
\begin{equation}
	z_2=-k_c\lceil z_1\rfloor^p-2\epsilon_y
\end{equation}
\begin{equation}
\left\{
\begin{aligned}
	&p<1,&0\leq |z_1|<1\\
	&p=1,&|z_1|=1\\
	&p>1,&|z_1|>1\\
\end{aligned}
\right.
\label{eq_dz2}
\end{equation}
where $0<r_b<1<r_1<r_t$. Select the negative feedback function $h_{m1}$ based on SSMD \eqref{eq_SSMD} with \eqref{eq_fxe}.

Under LoMC,
the controller \eqref{eq_uOF} can be simplified as
\begin{equation}
\begin{aligned}
v&=-k_u\ln\left(\frac{1+\widetilde{\xi}}{1-\widetilde{\xi}}\right)\\
&=-k_u\ln\left(\frac{1+(2\frac{s-\widetilde{y}_L}{\widetilde{y}_U-\widetilde{y}_L}-1)}{1-(2\frac{s-\widetilde{y}_L}{\widetilde{y}_U-y_L}-1)}\right)\\
&=-k_u\ln\left(\frac{s-\widetilde{y}_L}{\widetilde{y}_U-s}\right)
\end{aligned}
\label{eq_us}
\end{equation}
while $\widetilde{y}_L$ and $\widetilde{y}_U$ can be simplified as
\begin{equation}
\begin{aligned}
&\widetilde{y}_U=-\widetilde{y}_L\\
=&(1-\mathcal{T}_y)y_U+\mathcal{T}_y(|h_{m(n-1)}(s_{n-1})-\dot{s}_{n-1}|+\rho_e)\\
=&(1-\mathcal{T}_y)y_U+\mathcal{T}_y(|s|+\rho_e)
\end{aligned}
\label{eq_WyUyL}
\end{equation}
where $\mathcal{T}_y=\mathcal{S}\left(\frac{|s|-(y_U-\rho_e)}{\rho_e}\right)$.

\begin{figure}[]
	\centering
	\begin{overpic}[width=0.8\columnwidth]{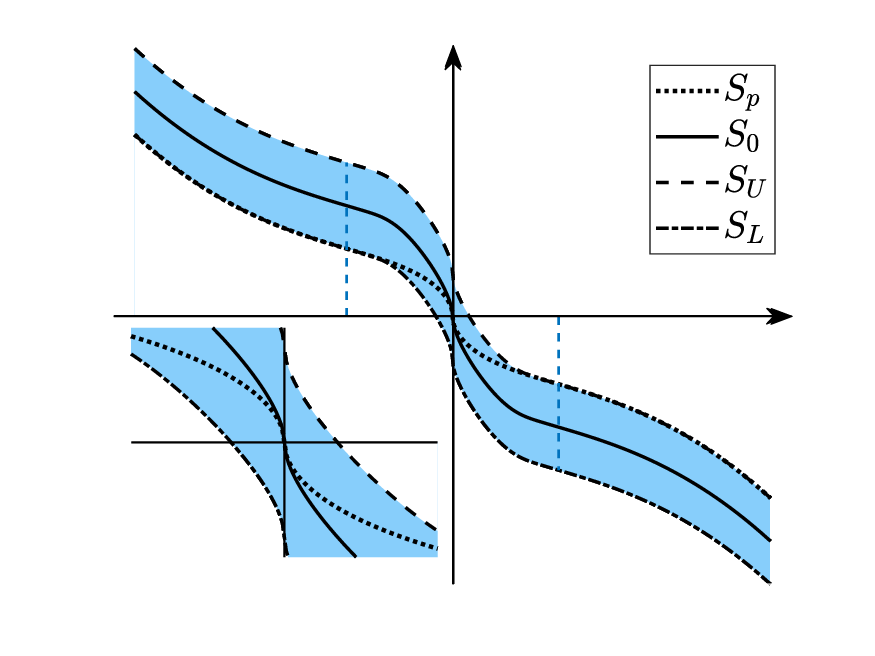}
		\put(88,35){$z_1$}
		\put(46,68){$z_2$}
		\put(36,35){$-\epsilon_{z1}$}
		\put(62,35){$\epsilon_{z1}$}
		\put(53,45){$\epsilon_y$}
		\put(44,29){$-\epsilon_y$}
	\end{overpic}
	\caption{LoMC and SSMD based variable exponent coefficients fixed-time manifold constraint diagram after $T_s$.}
	\label{fig_FixedeS}
\end{figure}

\begin{corollary}
	\label{corollary_fixedve}
	For the system \eqref{eq_SFS} with $n=2$ under Assumptions \ref{Assumption yd}-\ref{Assumption sat}, if the high gain differentiator \eqref{eq_HGD} and the controller \eqref{eq_us} (which is equivalent to \eqref{eq_uOF} composed of \eqref{eq_IS}, \eqref{eq_SSMD}, \eqref{eq_yUyL}, \eqref{eq_xi}, and \eqref{eq_fxe}) is used with parameters selection as
	\begin{equation}
	\begin{aligned}
	\epsilon_{s1}=\epsilon_y\\
	\epsilon_{z1}=\epsilon_z
	\end{aligned}
	\end{equation}
	in \eqref{eq_yUyL}, \eqref{eq_SSMD} and sufficiently small $\rho_e,\mu$, the close-loop system has the following properties:
	\begin{enumerate}[(\romannumeral1)]
		\item When the actuator output capability is sufficient, the setting time $T_{fxe}$ for $z_1$ to converge to the accuracy $|z_1|<\epsilon_z$ satisfies the following inequality:
		\begin{equation}
		T_{fxe}< T_s+T_{fxe}'.
		\label{eq_fxeT}
		\end{equation}
		with $T_{fxe}'=\frac{1}{k_c(r_1-1)}+\frac{1}{k_ce^{-\frac{a}{2e}}(1-r_b)}$;
		\item When the actuator output capability is insufficient, the system achieves a flexible control accuracy $|z_1|<\widetilde{x}_e$ within the same fixed time $T_{fxe}$, and after the controller exits saturation, it recovers to the preset accuracy $\epsilon_z$ within a finite time $\left(\max\{T_e,T_s\}+T_{fxe}'\right)$.
	\end{enumerate}
\end{corollary}

\begin{figure}[]
	\centering
	\begin{overpic}[width=0.8\columnwidth]{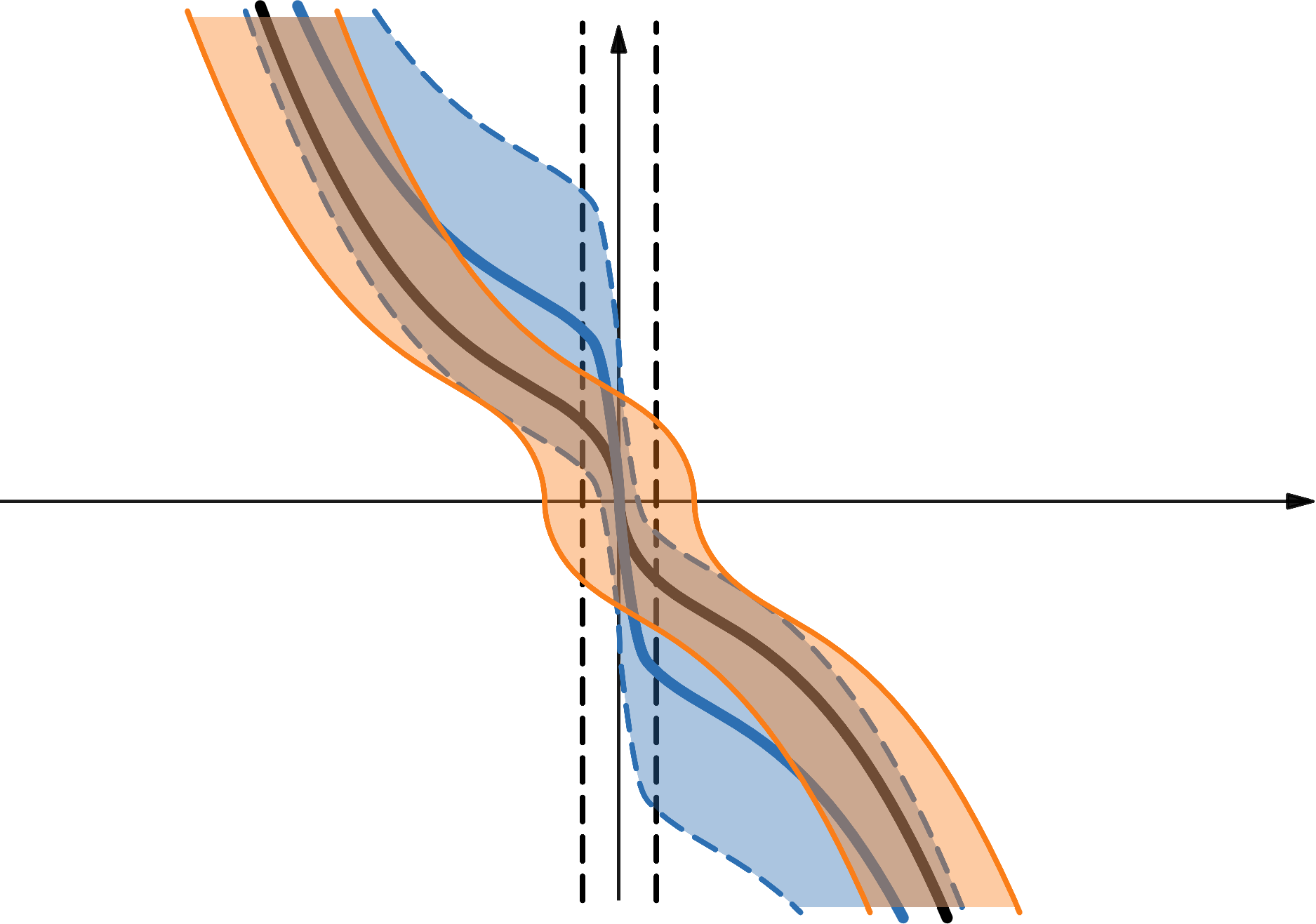}
		\put(95,35){$z_1$}
		\put(45,70){$z_2$}
		\put(53,33){$\widetilde{x}_e$}
	\end{overpic}
	\caption{A schematic diagram of equivalent longitudinal flexible expansion achieved through lateral flexible expansion. The thick black curve represents $z_2=h_{p1}(z_1)$, the black dashed line represents $z_1=\pm\epsilon_{z1}$, and the thick blue solid line represents $z_2=h_{m1}(z_1)$ with SSMD \eqref{eq_SSMD}. The orange area is formed by the lateral expansion of \(z_1=\pm\epsilon_{z1}\), while the blue area is the result of the longitudinal elastic expansion of the manifold \( z_2=h_{m1}(z_1) \) into \( \widetilde{\psi}_s \).}
	\label{fig_FxeProof}
\end{figure}

\textbf{Proof:}
According to \textbf{Theorem} \ref{theoremOF}, if the controller is used with non-activated saturation, $\textbf{Z}$ always stays within $\Psi_s$, i.e., the blue area in the Fig. \ref{fig_FixedeS}. The $\epsilon_{s1}$ in \eqref{eq_SSMD} is selected to be equal to the $\epsilon_y$ in \eqref{eq_yUyL}.
Then, $S_L$ will overlap with $S_p$ for $z_1<-\epsilon_{z1}$, and $S_U$ will overlap with $S_p$ for $z_1>\epsilon_{z1}$ after $T_s$ as shown in Fig. \ref{fig_FixedeS}. Therefore, the following inequality holds:
\begin{equation}
\left\{
\begin{aligned}
&z_2>k_c|z_1|^p,&z_1<-\epsilon_{z1}\\
&z_2<-k_c z_1^p,&z_1>\epsilon_{z1}
\end{aligned}
\right.
\label{eq_fxedz}
\end{equation}

The maximum time $T_{max}$ required for a system $z_2=\dot{z}_1=-k_c\lceil z\rfloor^p$ to converge to $0$ is derived as below.

For $z_1\geq0$, it holds that
\begin{equation}
\left\{
\begin{aligned}
&z_2\leq -k_c z_1^{r_1},&z_1\geq1\\
&z_2\leq -k_c z_1^{az_1^2+r_b}\leq-k_c e^{-\frac{a}{2e}}z_1^{r_b},&0\leq z_1<1
\end{aligned}
\right.
\label{eq_dz2}
\end{equation}
where $\min\{z_1^{z_1^2}\}=e^{-\frac{1}{2e}}$. It leads to
\begin{equation}
z_1(t)\leq
\left\{
\begin{aligned}
&\left(z_1^{1-r_1}(0)-k_c(1-r_1)t\right)^{\frac{1}{1-r_1}},&t<t_1\\
&\left(z_1^{1-r_b}(t_1)-k_c(1-r_b)t\right)^{\frac{1}{1-r_b}},&t_1\leq t< t_2\\
&0,&t\geq t_2
\end{aligned}
\right.
\label{eq_z1eq}
\end{equation}
where $t_1$ and $t_2$ are the first moment of $z_1(t)=1$ and $z_1(0)=0$, respectively. $t_1\leq\frac{z_1^{1-r_1}(0)-1}{k_c(1-r_1)}=\frac{1}{k_c(r_1-1)}$, $t_2\leq t_1+\frac{z_1^{1-r_b}(t1)}{k_ce^{-\frac{a}{2e}}(1-r_b)}=t_1+\frac{1}{k_ce^{-\frac{a}{2e}}(1-r_b)}$. Therefore $T_{max}=\frac{1}{k_c(r_1-1)}+\frac{1}{k_ce^{-\frac{a}{2e}}(1-r_b)}$.
The situation for $z_1\leq0$ is same as above.

By comparison, the system composed of inequality \eqref{eq_fxedz} has a faster convergence speed when $|z_1|>\epsilon_{z1}$, so the convergence time to $|z_1|\leq\epsilon_{z1}$ is less than $T_{max}$.
The final accuracy is determined by $\epsilon_{z1}$ which can be set as $\epsilon_z$.

When the controller becomes saturated and the system state will remain within the flexible constraint set $\widetilde{\Psi}_s$ \eqref{eq_WPsis}. Since $p\geq r_1>1$ for $|s_1|\geq1$ in \eqref{eq_fxe}, there must exist a lateral flexible expansion set for non skewed manifolds (as shown in the orange area of Fig. \ref{fig_FxeProof})
\begin{equation}
\widetilde{\Psi}_{se}:=\{\textbf{Z}|h_{p1}(z_1+\widetilde{x}_e)<z_2<h_{p1}(z_1-\widetilde{x}_e)\}
\label{eq_WPsisFxe}
\end{equation}
which achieves slightly lower control accuracy $\widetilde{x}_e$ than the flexible constraint set $\widetilde{\Psi}_s=\{\textbf{Z}|h_{m1}(z_1)+\widetilde{y}_L<z_2<h_{m1}(z_1)+\widetilde{y}_U\}$ formed for longitudinal flexible expansion of skewed manifolds (as shown in the blue area of Fig. \ref{fig_FxeProof}) at a fixed time. The manifold \( h_{m1}(z_1)+ \widetilde{y}_U\) and \( h_{p1}(z_1-\widetilde{x}_e) \) intersect at only one point in the fourth quadrant.
The curves \(z_2=h_{m1}(z_1)+ \widetilde{y}_U\) and \(z_2=h_{p1}(z_1-\widetilde{x}_e)\) intersect at only one point in the fourth quadrant. The shapes of the two curves, \( \frac{\partial h_{p1}(0)}{\partial z_1} = -\infty \) and \( -\infty<h_{m1}(z_1) < 0 \forall z_1\neq0\), ensure that there is always a solution that guaranteeing a unique intersection in the fourth quadrant as Fig. \ref{fig_FxeProof}.

Similar to the previous analysis in proof of corollary \ref{corollary_finite}, the flexible accuracy is $\widetilde{x}_e$. After exiting saturation, if the exit time $T_e$ is earlier than the preset time $T_s$, the time to converge to the preset accuracy $\epsilon_{z}$ remains $T_{fxe}$. If the exit time is later than the preset time $T_e\geq T_s$, the time to restore the preset accuracy will not exceed \(\left(T_e+T_{fxe}'\right)\).
$\hfill\blacksquare$

\begin{remark}
	The variable exponential coefficient function proposed in the \cite{moulay2021robust} and \cite{moulay2022fixed} is
	\begin{subequations}
		\begin{align}
		&h_p(z)=-\beta\lceil z\rfloor^p\\
		&p(z)=\frac{\lambda_1 z^2}{1+\mu_1 z^2},\theta_1=\frac{\lambda_1}{1+\mu_1}>1,\beta>0
		\end{align}
		\label{eq_fvec}
	\end{subequations}
	which has an exponential coefficient of $0$ when $z=0$, resulting in the value of $h_p(0)$ is $-1$ and the left limit is $\lim_{z\rightarrow0^-}h_p(z)=1$. Therefore, the sliding surface is discontinuous \cite{su2023comments}. There is a jumping feedback form at the equilibrium point, with similar behavior as a sign term. In fact, if the exponential coefficient $p(z)$ is set as $p(z)=0$, it is the sign function. Therefore, the controller's chattering can be seen in the simulation example of \cite{moulay2021robust} and \cite{moulay2022fixed}. The design in \eqref{eq_fxe} improves the variable exponential coefficient function by adjusting the exponential coefficients of the power feedback function at $z=0$, $|z|=1$, and $|z|=\infty$ using three parameters $r_b$, $r_1$, and $r_b$, respectively, while ensuring that the lowest exponential coefficients $r_b$ is greater than $0$, thus guaranteeing the continuity of the function.
\end{remark}

\begin{figure}[]
	\centering
	\subfigure[]
	{
		\begin{overpic}[width=0.45\columnwidth]{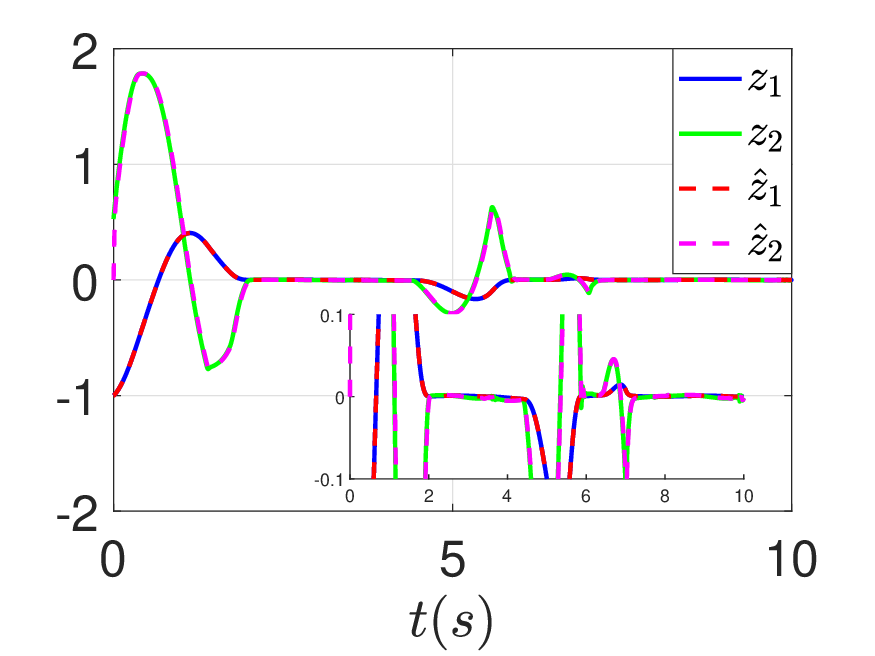}
		\end{overpic}
		\label{fig_FixedveSz}
	}\hfill
	\subfigure[]
	{
		\begin{overpic}[width=0.45\columnwidth]{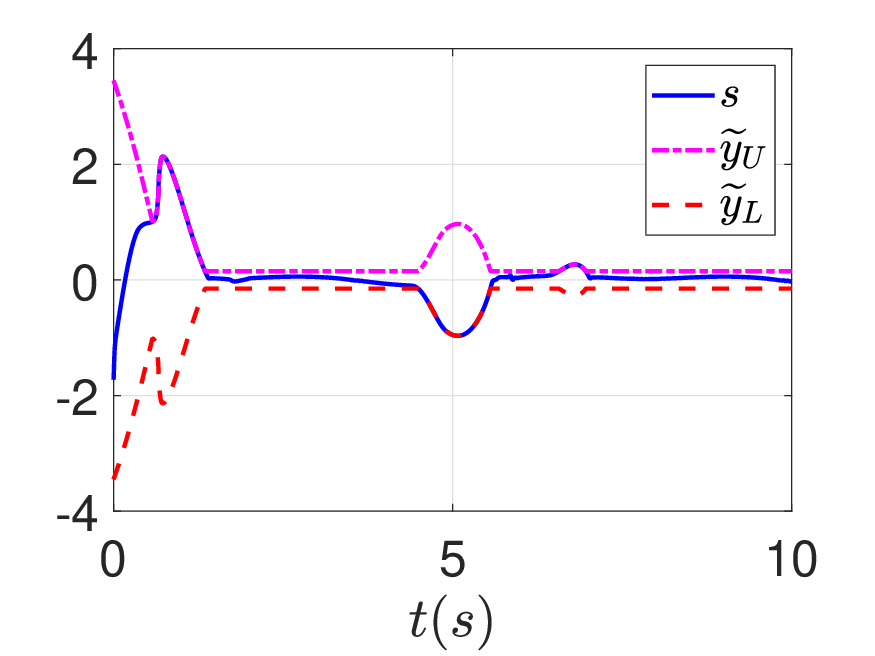}
		\end{overpic}
		\label{fig_FixedveSs}
	}\\
	\subfigure[]
	{
		\begin{overpic}[width=0.45\columnwidth]{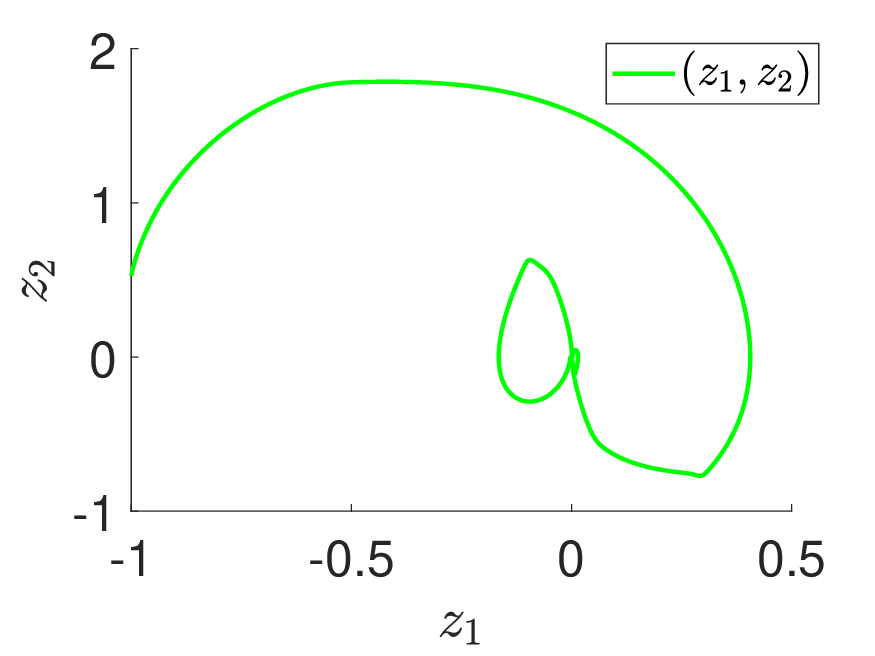}
		\end{overpic}
		\label{fig_FixedveSm}
	}\hfill
	\subfigure[]
	{
		\begin{overpic}[width=0.45\columnwidth]{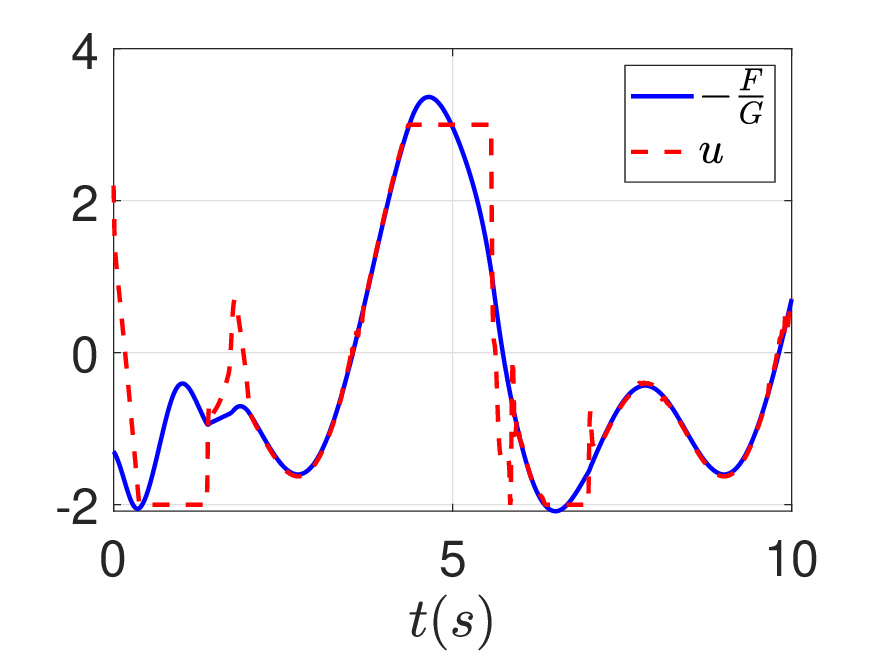}
		\end{overpic}
		\label{fig_FixedveSu}
	}
	\caption{Simulation results of the fixed-time PPC with \eqref{eq_fxe}, LoMC, SSMD, and saturated actuator. (a) The variables $z_1$, $z_2$, $\hat{z}_1$, and $\hat{z}_2$ versus time. (b) The variables $s$, $\widetilde{y}_U$, and $\widetilde{y}_L$ versus time. (c) Phase trajectory $(z_1,z_2)$. (d) The evolution of $-\frac{F}{G}$ and $u$.}
	\label{fig_FixedveSat}
\end{figure}

The simulation system is same as \eqref{eq_System_sim} with input constraint parameters \( \underline{u}=-2,\bar{u}=3 \). The controller parameters are selected as follows: $k_u=2,k_0=2,\epsilon_{s1}=\epsilon_y=0.15,T_s=1,\epsilon_{z1}=\epsilon_z=0.1,k_{pp1}=1,\rho_e=0.01$ in controller, $k_c=1.5,r_b=0.5,r_1=2,r_t=3$ in \eqref{eq_fxe}, and $a_1=4,a_2=4,\mu=0.01$ in differentiator \eqref{eq_HGD}, therefore $T_{fxe}=4.3$.

The simulation results are shown in Fig. \ref{fig_hoFixedSat}. The actuator output \( u \) remains within the input constraint as Fig. \ref{fig_FixedveSu}. The flexible constraint $\widetilde{y}_U,\widetilde{y}_L$ expands as \( s \) exceeds the original constraint $y_U,y_L$ as Fig. \ref{fig_FixedveSs}, after exiting saturation, it restores to its original preset value $y_U,y_L$. The control accuracy is also able to recover to the preset accuracy within a fixed time as Fig. \ref{fig_FixedveSz}.

\section{High-order Systems Control}
\label{High-order Systems}

\subsection{Linear Fully Actuated Manifold Constraint Control}
\label{subs_SS_LLFMC}

Linear feedback law is selected as negative feedback function $h_{mi}$:
\begin{equation}
	h_{mi}(s_i)=-b_is_i
	\label{eq_lfmi}
\end{equation}
The manifold constructed based on iterative methods \eqref{eq_IS} is:
\begin{equation}
	s=\prod_{i=1}^{n-1}\left(\frac{\partial}{\partial t}+b_i\right)s_1
	\label{eq_LFAM}
\end{equation}
The polynomial \eqref{eq_LFAM} has $n-1$ negative real poles $-b_i,i=1,2,\cdots,n-1$, which obviously satisfies the Hurwitz condition. It is consistent with the designs in many literature \cite{bechlioulis2013output,wei2018robust,dimanidis2020output,lv2023distributed,song2016adaptive,cao2020adaptive}. The linear manifold is a straight line on the $(s_{n-1},\dot{s}_n-1)$ phase plane, so the effects of OMC, LaMC, and LoMC are the same.
According to equation \eqref{eq_us}, the constraint on the linear manifold is ultimately a direct constraint on the value $s$.
Therefore, LMCC \cite{bechlioulis2013output,wei2018robust,dimanidis2020output,lv2023distributed} are special cases of the proposed FAMCC with linear manifold.

\subsection{FAMCC Based Recursive Fixed-time Control (RFC)}
\label{subs_SS_RFC}

An iterative fixed-time fully actuated manifold can be constructed by \eqref{eq_IS}, NSMD \eqref{eq_NSMD}, LoMC \eqref{eq_yUyL}, and fixed-time feedback law as follows:
\begin{equation}
	h_{pi}(s_i)=-\alpha_i\lceil s_i\rfloor^{p_i}-\beta_i\lceil s_i\rfloor^{q_i},i=1,\dots,n-1
	\label{eq_RFCs}
\end{equation}
where $p_i$ and $q_i$ are positive constants which satisfy $0<{p_i}<1<{q_i}$.


\begin{corollary}
	\label{corollary_fixedho}
	For the system \eqref{eq_SFS}, under Assumptions \ref{Assumption yd}-\ref{Assumption sat}, if the high gain differentiator \eqref{eq_HGD} and the recursive fixed-time controller composed of \eqref{eq_IS}, \eqref{eq_NSMD}, \eqref{eq_yUyL}, \eqref{eq_RFCs} and \eqref{eq_uOF} (which also can be simplified as \eqref{eq_us} under LoMC) is applied with parameters selection as:
	\begin{subequations}
		\begin{align}
		&\epsilon_{zj}=\epsilon_{s(j-1)},\quad j=2,\cdots,n-1
		\label{eq_NSMDie}\\
		&\epsilon_{z1}=\epsilon_z,\epsilon_{s(n-1)}=\epsilon_y
		\label{eq_NSMDif}
		\end{align}
		\label{eq_NSMDi}
	\end{subequations}
	in \eqref{eq_yUyL}, \eqref{eq_SSMD} and sufficiently small $\rho_e,\mu$, the close-loop system has the following properties:
	\begin{enumerate}[(\romannumeral1)]
		\item When the actuator output capability is sufficient, the setting time $T_{fr}$ for $z_1$ to converge to the accuracy $|z_1|<\epsilon_z$ satisfies the following inequality:
		\begin{equation}
		T_{fr}< T_s+T_{fr}'
		\label{eq_fxTho}
		\end{equation}
		with $T_{fr}'=\sum_{i=1}^{n-1}\left(\frac{1}{\alpha_i(1-p_i)}+\frac{1}{\beta_i(q_i-1)}\right)$;
		\item When the actuator output capability is insufficient, the system achieves a flexible control accuracy $|s_j|<\widetilde{x}_{ej}$ and $|z_1|<\widetilde{x}_{e1}$
		\begin{equation}
		\begin{aligned}
		&\widetilde{x}_{e(n-1)}=\max\{\epsilon_{z(n-1)},\widetilde{y}_U-y_U\},\\ &\widetilde{x}_{ej}=\max\{\epsilon_{zj},\widetilde{x}_{e(j+1)}\}.
		\end{aligned}
		\label{eq_Wxei}
		\end{equation}
		within the same fixed time $T_{fr}$, and after the controller exits saturation, it recovers to the preset accuracy $\epsilon_z$ within a finite time $\left(\max\{T_e,T_s\}+T_{fr}'\right)$ where $T_e$ is the time when the system exits saturation.
	\end{enumerate}
\end{corollary}

\textbf{Proof:}
Define manifold constrain region of $(s_i,\dot{s}_i)$ with $i=1,2,\cdots,n-1$ as follows:
\begin{equation}
\Psi_{si}:=\{(s_i,\dot{s}_i)|-\epsilon_{si}<\dot{s}_i-h_{mi}(s_i)<\epsilon_{si}\}
\label{eq_Psisi}
\end{equation}
where $\Psi_{s(n-1)}=\Psi_s$ in \eqref{eq_Psis} after $T_s$ since $\epsilon_{s(n-1)}=\epsilon_y$.

When the actuator output capability is sufficient, $u=v$ can be guaranteed during the operation. According to \textbf{Theorem} \ref{theoremOF}, if the controller
is used with parameter selection \eqref{eq_NSMDi}, $\textbf{Z}$ always stays within $\Psi_s$, i.e., $|s|=|s_n|<\epsilon_{n-1}=\epsilon_y$ and $(s_{n-1},\dot{s}_{n-1})$ always in $\Psi_{s(n-1)}$ after $T_s$.
It can be obtained that the following inequality holds:
\begin{equation}
\left\{
\begin{aligned}
&\dot{s}_i>\alpha_i|s_i|^{p_i}+\beta_i|s_i|^{q_i},&s_i<-\epsilon_{zi}\\
&\dot{s}_i<-\alpha_i s_i^{p_i}-\beta_i s_i^{q_i},&s_i>\epsilon_{zi}
\end{aligned}
\right.
\label{eq_fxdzho}
\end{equation}
for $i=n-1$.
Referring to the fixed time control theorem \cite{ni2017fixed}, since $s_{n-1}$ has faster converge speed than fixed-time system form \eqref{eq_fxdzho}, $s_{n-1}$ will converge to $|s_{n-1}|<\epsilon_{z(n-1)}=\epsilon_{s(n-2)}$ within a fixed time $T_s+\frac{1}{\alpha_{n-1}(1-p_{n-1})}+\frac{1}{\beta_{n-1}(q_{n-1}-1)}$.

For $i=n-2,\dots,1$, it can be recursively obtained that the inequality \eqref{eq_fxdzho} holds when $|s_{i+1}|\leq\epsilon_{si}$.
Combining the fixed time control theorem, $s_i$ will converge to $|s_i|<\epsilon_{zi}$ within fixed time $T_s+\sum_{k=i}^{n-1}\left(\frac{1}{\alpha_k(1-p_k)}+\frac{1}{\beta_k(q_k-1)}\right)$.
Recurring from $i=n-2$ to $i=1$, it can be obtained that the final output accuracy is $|z_1|=|s_1|<\epsilon_{z1}$, and the convergence time $T_{fr}$ can be expressed as \eqref{eq_fxTho}, which is related to $T_s$ and the parameters of $h_{pi}$ in \eqref{eq_RFCs}.
The final accuracy is determined by $\epsilon_{z1}$ which can be set as $\epsilon_z$.

When the actuator output capability is insufficient, the system state will remain within the flexible constraint set $\widetilde{\Psi}_s$ \eqref{eq_WPsis} as the design in \eqref{eq_WxyUxyL}. For $i=1,\dots,n-1$, since $\frac{\partial{h}_{pi}(s_i)}{\partial s_i}<-1,\forall s_i\in \mathcal{R}$ in \eqref{eq_RFCs}, $\frac{\partial{h}_{mi}(s_i)}{\partial s_i}<-1, \forall |s_i|\geq \epsilon_{zi}$, the performance change caused by a lateral movement of a certain distance of \( h_{pi} \) is greater than the performance change caused by the same distance of longitudinal movement.
Thus, the performance of the flexible constraint set $\widetilde{\Psi}_s$, formed by the longitudinal flexible expansion of \( h_{mi} \) with NSMD \eqref{eq_NSMD}, can be approximated by the constraint set $\widetilde{\Psi}_{si}$, formed by the lateral flexible expansion of \( h_{pi} \) as
\begin{equation}
\widetilde{\Psi}_{si}:=\{(s_i,\dot{s}_i)|h_{pi}(s_i+\widetilde{x}_{ei})<\dot{s}_i<h_{pi}(s_i-\widetilde{x}_{ei})\}
\label{eq_WPsisFxi}
\end{equation}
where $\widetilde{x}_{ei}$ is defined in \eqref{eq_Wxei}.
Similar to the previous analysis in proof of corollary \ref{corollary_finite}, the flexible accuracy is denoted as $|s_i|\leq\widetilde{x}_{ei}$.
Recursively, the final flexible accuracy is obtained as \(\widetilde{x}_{e1}\).

After exiting saturation, if the exit time $T_e$ is earlier than the preset time $T_s$, the time to converge to the preset accuracy $\epsilon_{z}$ remains $T_{fr}$. If the exit time is later than the preset time $T_e\geq T_s$, the time to restore the preset accuracy will not exceed \(\left(T_e+T_{fr}'\right)\).
$\hfill\blacksquare$


\begin{figure}[]
	\centering
	\subfigure[]
	{
		\begin{overpic}[width=0.45\columnwidth]{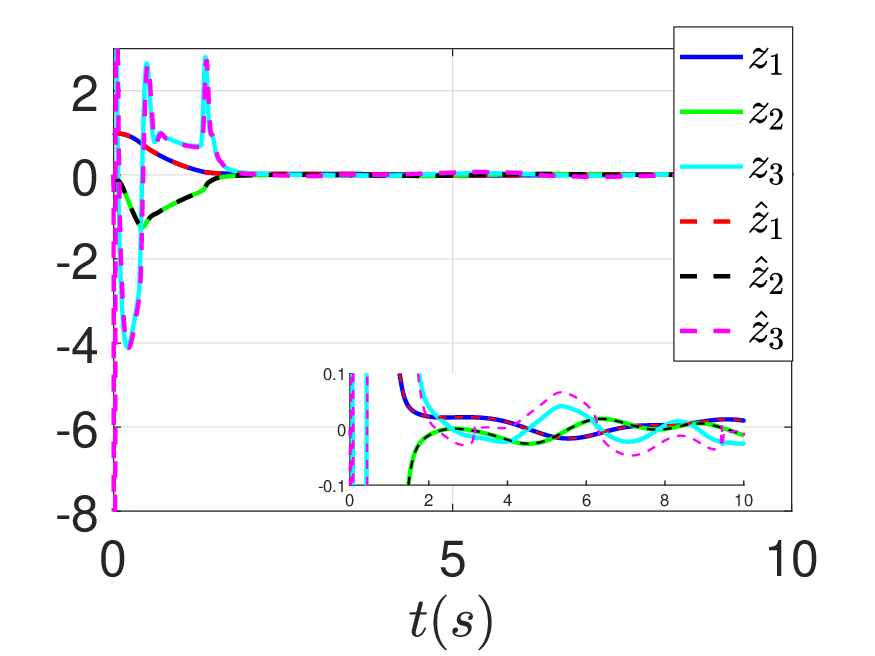}
		\end{overpic}
		\label{fig_hoFixedz}
	}\hfill
	\subfigure[]
	{
		\begin{overpic}[width=0.45\columnwidth]{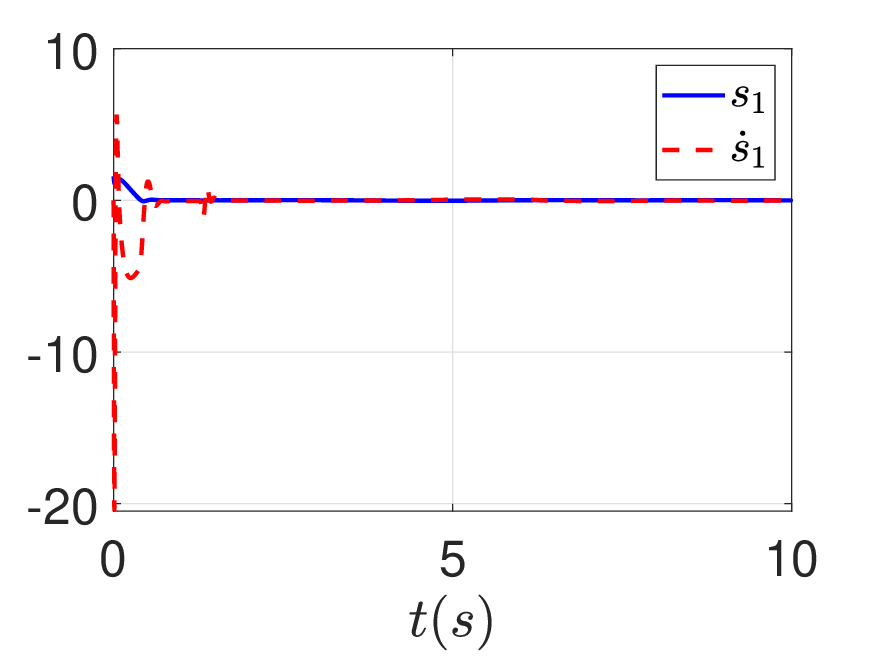}
		\end{overpic}
		\label{fig_hoFixeds}
	}\\
	\subfigure[]
	{
		\begin{overpic}[width=0.45\columnwidth]{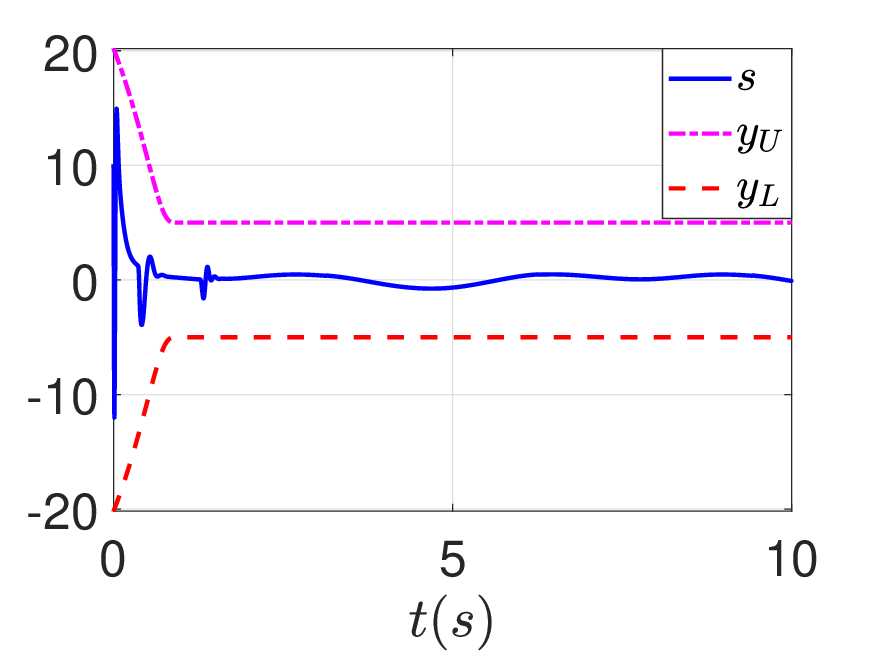}
		\end{overpic}
		\label{fig_hoFixedrhos}
	}\hfill
	\subfigure[]
	{
		\begin{overpic}[width=0.45\columnwidth]{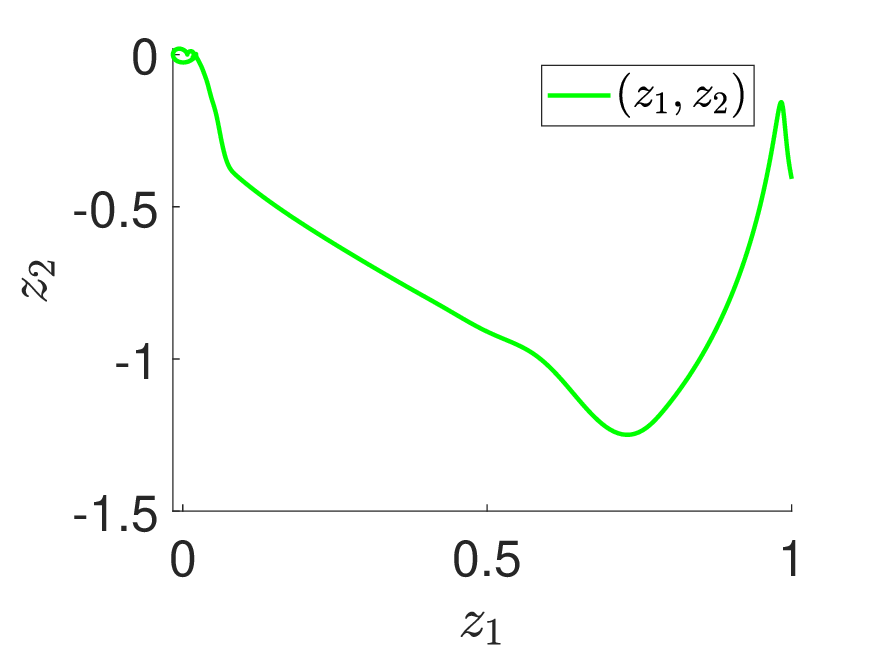}
		\end{overpic}
		\label{fig_hoFixedm}
	}\\
	\subfigure[]
	{
		\begin{overpic}[width=0.45\columnwidth]{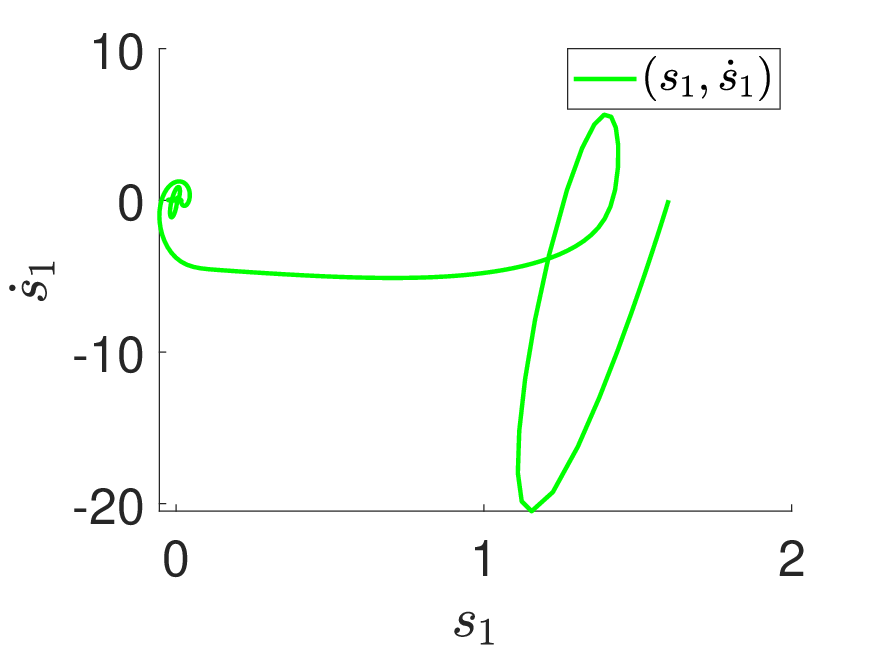}
		\end{overpic}
		\label{fig_hoFixedms}
	}\hfill
	\subfigure[]
	{
		\begin{overpic}[width=0.45\columnwidth]{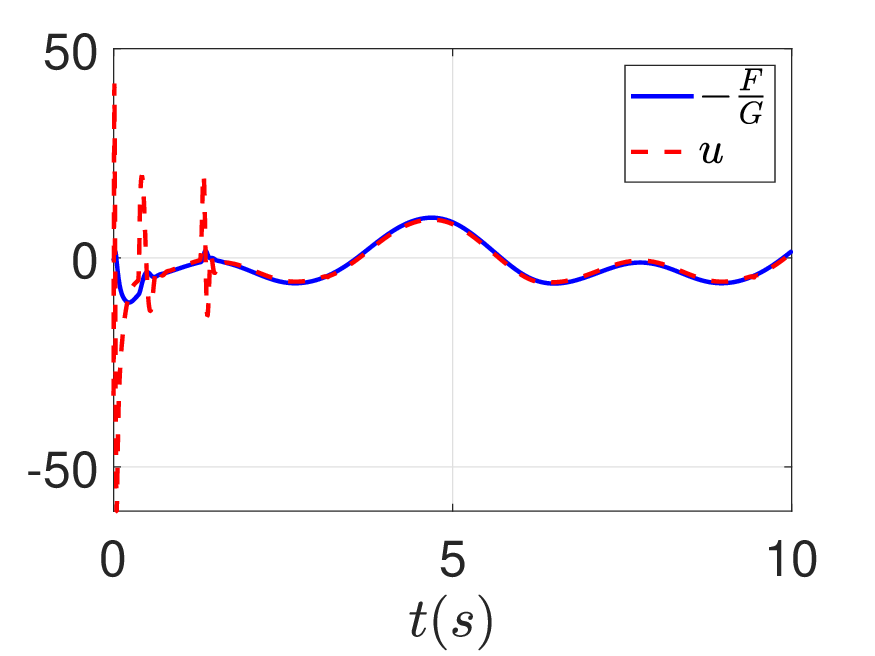}
		\end{overpic}
		\label{fig_hoFixedu}
	}
	\caption{Simulation results of the RFC. (a) The variables $z_1$, $z_2$, $z_3$, $\hat{z}_1$, $\hat{z}_2$ and $\hat{z}_3$ versus time. (b) The variables $s,\dot{s}_1$ versus time (c) The variables $s$, $y_U$, and $y_L$ versus time. (d) Phase trajectory $(z_1,z_2)$. (e) Phase trajectory $(s_1,\dot{s}_1)$. (f) The evolution of $-\frac{F}{G}$ and $u$.}
	\label{fig_hoFixed}
\end{figure}

\begin{figure}[]
	\centering
	\subfigure[]
	{
		\begin{overpic}[width=0.45\columnwidth]{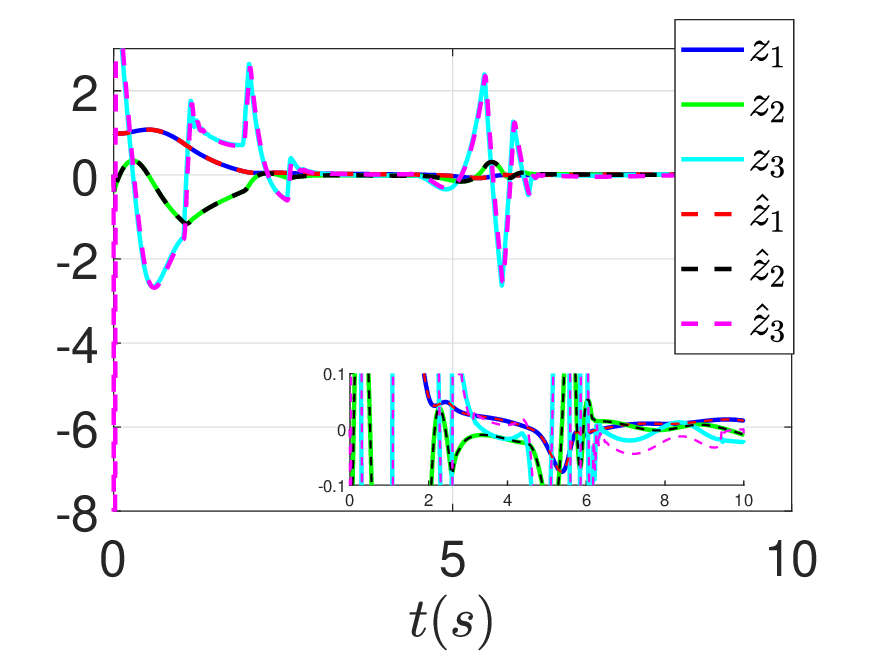}
		\end{overpic}
		\label{fig_hoFixedSz}
	}\hfill
	\subfigure[]
	{
		\begin{overpic}[width=0.45\columnwidth]{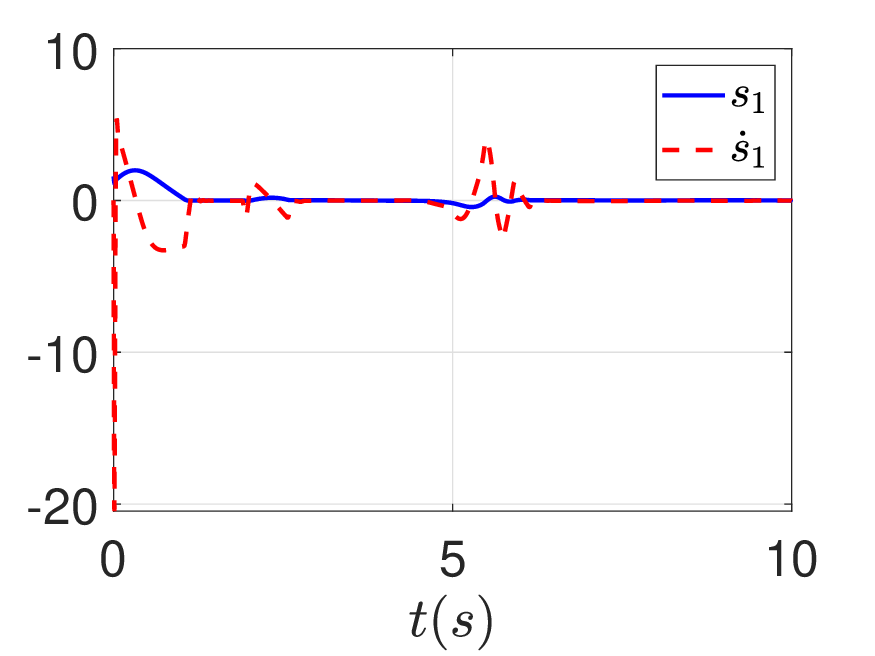}
		\end{overpic}
		\label{fig_hoFixedSs}
	}\\
	\subfigure[]
	{
		\begin{overpic}[width=0.45\columnwidth]{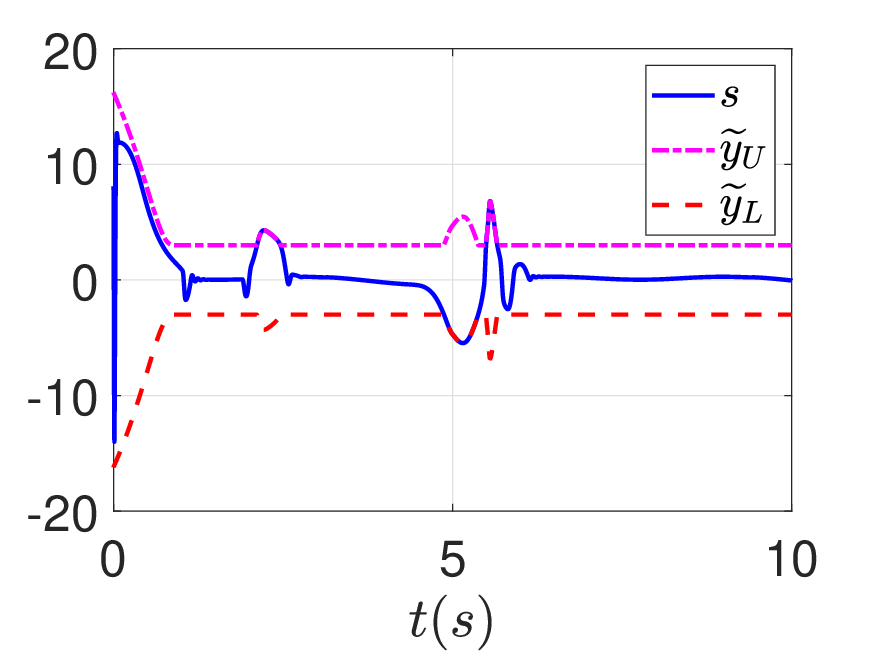}
		\end{overpic}
		\label{fig_hoFixedSrhos}
	}\hfill
	\subfigure[]
	{
		\begin{overpic}[width=0.45\columnwidth]{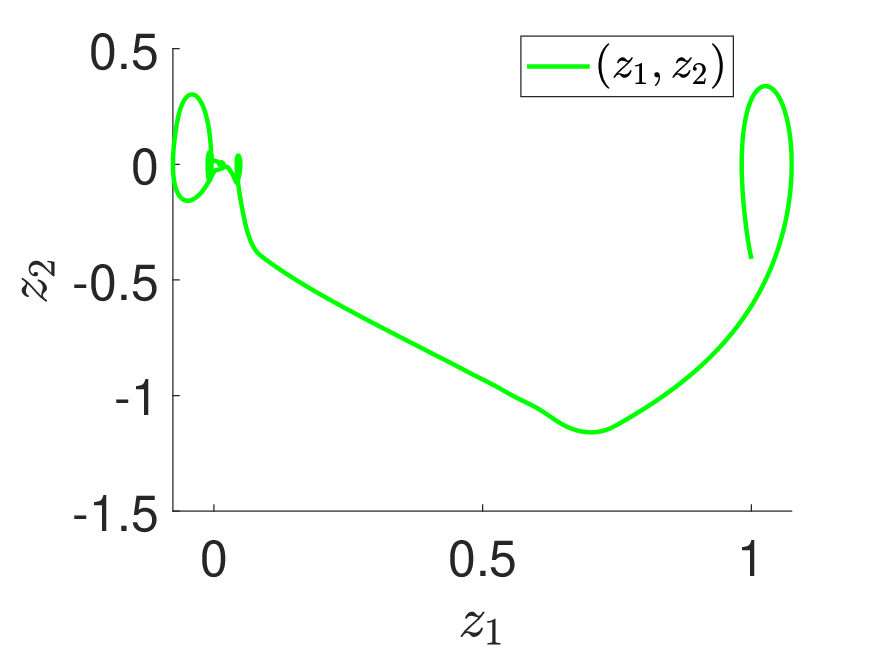}
		\end{overpic}
		\label{fig_hoFixedSm}
	}\\
	\subfigure[]
	{
		\begin{overpic}[width=0.45\columnwidth]{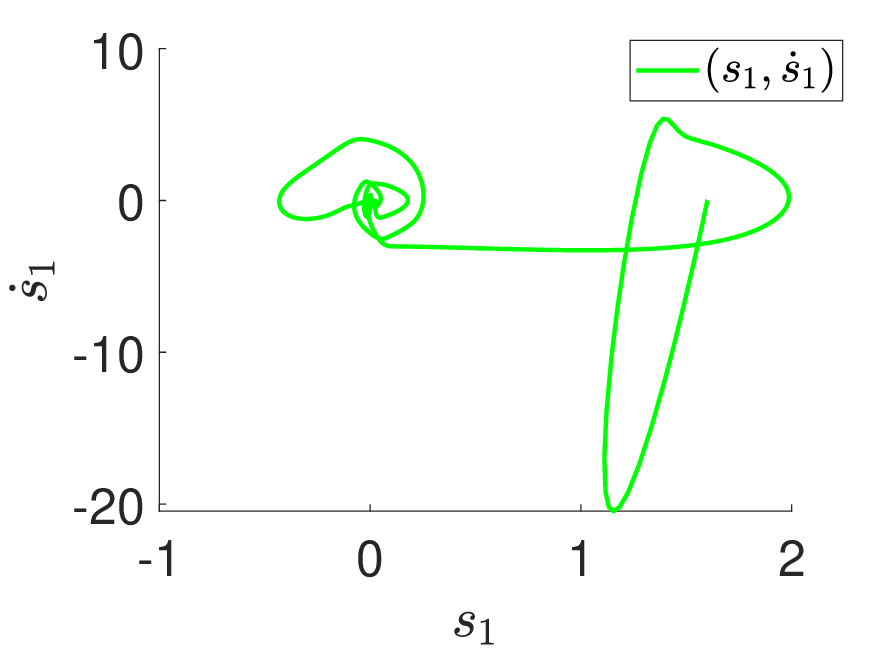}
		\end{overpic}
		\label{fig_hoFixedSms}
	}\hfill
	\subfigure[]
	{
		\begin{overpic}[width=0.45\columnwidth]{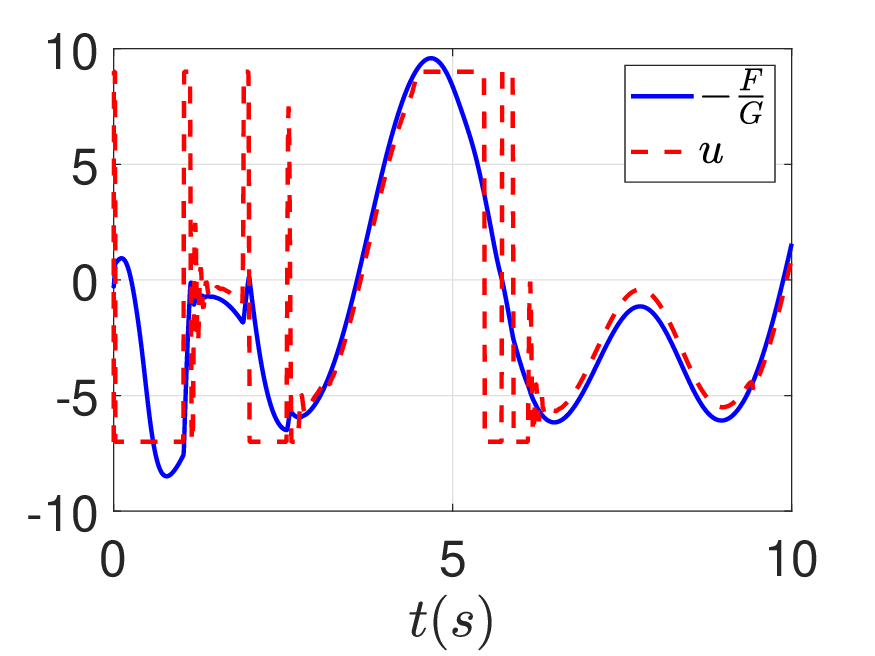}
		\end{overpic}
		\label{fig_hoFixedSu}
	}
	\caption{Simulation results of the RFC with saturated actuator. (a) The variables $z_1$, $z_2$, $z_3$, $\hat{z}_1$, $\hat{z}_2$ and $\hat{z}_3$ versus time. (b) The variables $s,\dot{s}_1$ versus time (c) The variables $s$, $\widetilde{y}_U$, and $\widetilde{y}_L$ versus time. (d) Phase trajectory $(z_1,z_2)$. (e) Phase trajectory $(s_1,\dot{s}_1)$. (f) The evolution of $-\frac{F}{G}$ and $u$.}
	\label{fig_hoFixedSat}
\end{figure}

To verify the effectiveness of the proposed FAMCC based on RFC, the following system is used in the simulation:
\begin{equation}
\left\{
\begin{aligned}
\dot{x}_1=&(1+0.1\sin(x_1)+0.1\cos(t))-x_1+cos(t),\\
\dot{x}_2=&(1+0.1\cos(x_1x_2)+0.1\sin(t))x_3\\&-(x_1^2-1)x_2+\sin(t),\\
\dot{x}_3=&(2+0.2\sin(x_1x_2x_3)+0.1\cos(t))u(v)\\&-(x_1^2+x_2)x_3+10\sin(t)+10\cos(2t)\\
y=&x_1,\underline{u}=-100,\bar{u}=100,\\
x_1(&0)=1,x_2(0)=-0.2,x_3(0)=0.4\\
y_d(&t)=sin(t).\\
\end{aligned}
\right.
\label{eq_System_sim3}
\end{equation}
The controller parameters are selected as follows: $k_u=30,k_0=2,\epsilon_{z1}=\epsilon_z=0.1,\epsilon_{s1}=\epsilon_{z2}=0.2,\epsilon_{s2}=\epsilon_y=5,T_s=1,k_{pp1}=0.1,k_{pp2}=1,\rho_e=0.01$ in controller, $p_1=p_2=0.5,q_1=q_2=2,\alpha_1=1,\beta_1=0.5,\alpha_2=2,\beta_2=1$ in \eqref{eq_RFCs}, and $a_1=4,a_2=6,a_3=4,\mu=0.01$ in differentiator \eqref{eq_HGD}, therefore $T_{fx}<7$.

The simulation results are shown in Fig. \ref{fig_hoFixedz}-Fig. \ref{fig_hoFixedu}. It can be seen that $s$ always evolves within the prescribed bounds $y_U$ and $y_L$, and $z_1$ achieves the prescribed accuracy of $0.1$ within the settling time of $1.2s<T_{fr}$.

If the input constraint in the system \eqref{eq_System_sim3} is reduced to \( \underline{u}=-7,\bar{u}=9 \). The simulation results are shown in Fig. \ref{fig_hoFixedSat}. The actuator output \( u \) remains within the input constraint as Fig. \ref{fig_hoFixedSu}. The flexible constraint $\widetilde{y}_U,\widetilde{y}_L$ expands as \( s \) increases in Fig. \ref{fig_hoFixedSrhos}, after exiting saturation, it restores to its original preset value $y_U,y_L$. The control accuracy is also able to recover to the preset accuracy within a fixed time as Fig. \ref{fig_hoFixedz}.

The MATLAB simulation codes can be obtained at \url{https://github.com/Mudianrui/FAMCC-SAT}.

\section{Conclusion}



\begin{figure}[]
	\centering
	\begin{overpic}[width=0.8\columnwidth]{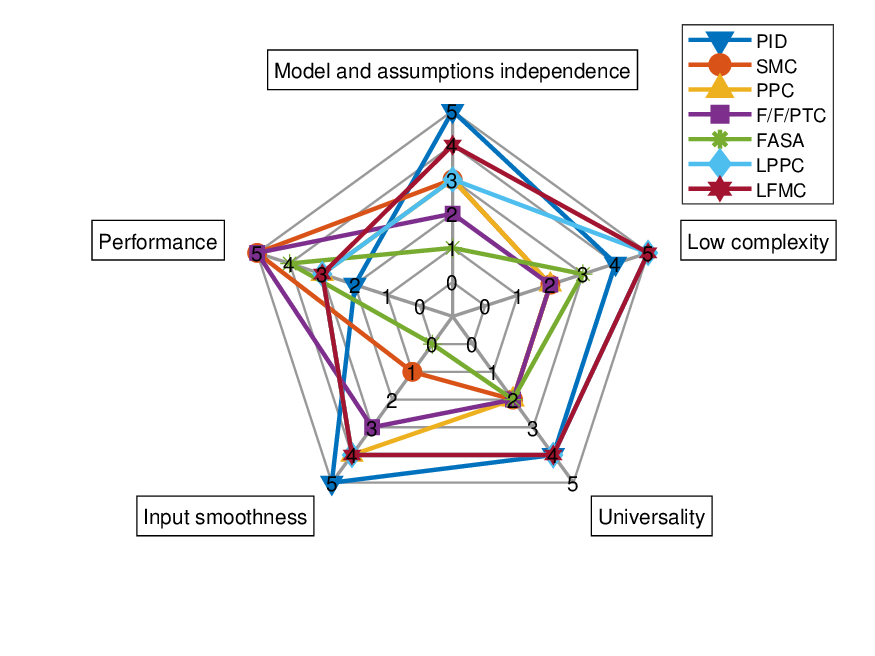}
	\end{overpic}
	\caption{Comparison of characteristics of different controllers.}
	\label{fig_controlSpider}
\end{figure}

\subsection{Characteristics of This Methods}

\begin{enumerate}[(\romannumeral1)]
	\item The linear manifold constraint control method \cite{bechlioulis2013output,wei2018robust,dimanidis2020output,lv2023distributed,song2016adaptive,cao2020adaptive} is extended to nonlinear manifold constraint so that the steady-state accuracy can be directly preset in the controller.
	\item The final approximation-free controller only uses the output and the order of the system without various limitations on the system as \cite{liu2022reduced,li2023adaptive,zhang2019k,wang2020global} and \cite{cai2023active}.
	\item In the simulation, the controller almost coincides with the normalized total disturbance $-\frac{F}{G}$, demonstrating the ability to approximate the optimal energy.
\end{enumerate}

\subsection{Extension on This Methods}
For the future researches, further improvements are required as follows:

\begin{enumerate}[(\romannumeral1)]
	\item Remove the assumption of model differentiability in Assumption \ref{Assumption gf}. This may require the use of nonsmooth dynamic theory based on Filippov.
	\item The proposed control method exhibits low complexity characteristics in the control of a second-order system, but as the order increases, the differential explosion may occur due to the iterative generation of the manifold. This may require finding a non-linear fast manifold construction method that does not rely on iteration.
\end{enumerate}

\appendix
\section{Proof of Lemma \ref{lemma1}}
\label{Proof_Lemma}

The fully actuated iterative manifold can be rewritten as:
\begin{equation}
s=\prod_{i=1}^{n-1}\left(\frac{\partial}{\partial t}-\frac{h_{mi}}{s_i}\right)s_1
\label{eq_Snon}
\end{equation}
According the definition of $h_{mi}$ in \eqref{eq_IS} and \textbf{Assumption} \ref{Assumption gf}, there must exit positive constants $r_i,i=1,\dots,n-1$ such that $s=0$ has $n-1$ negative real roots $\frac{h_{mi}}{s_i}<-r_i$,
and there is a positive constant $\bar{s}$ such that $|s|<\bar{s}$ since $\textbf{Z}$ can be maintained at the vicinity of the manifold $s=0$ when $|\xi(\textbf{Z})|<1$.

\begin{figure}[]
	\centering
	\begin{overpic}[width=1\columnwidth]{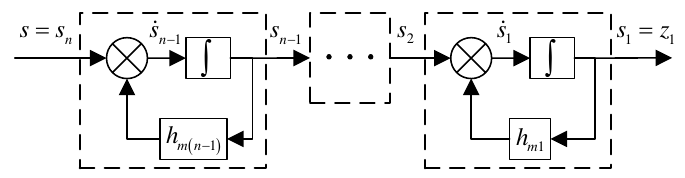}
	\end{overpic}
	\caption{Fully actuated iterative manifold diagram.}
	\label{fig_SF}
\end{figure}

The iterative manifold can be considered as a series of first order linear low pass filters as shown in Fig. \ref{fig_SF}.
Referring to the Proposition 2 in \cite{bechlioulis2013output}, it can be obtained that there exist positive constants $\bar{s}_{i}$ and $\bar{z}_i$ such that
\begin{equation}
|s_i|\leq\bar{s}_i+\frac{\bar{s}}{\prod_{j=i}^{n-1}r_j}
\label{eq_barSi}
\end{equation}
and
\begin{equation}
|z_i|\leq\bar{z}_i+\frac{2^{n-1-i}\bar{s}}{\prod_{j=i}^{n-1}r_j}
\label{eq_barzi}
\end{equation}
for all $t\geq0$, $i=1,\dots,n-1$.

\section{Proof of Theorem \ref{theoremSF}}
\label{Proof_theoremSF}

From Assumptions \ref{Assumption yd}-\ref{Assumption gf} and the design of the $x_U,x_L,y_U,y_L$ in \eqref{eq_xUxL} and \eqref{eq_yUyL}, it can be obtained that $\textbf{Z}(0)\in\Psi_s$, $\xi(0)\in(-1,1)$, and all signals are bounded at the beginning.

We can prove the theorem in two different situations as following:
situation A is $\frac{\partial{h_{m(n-1)}(\bullet)}}{\partial{\bullet}}\in L_{\infty}$,
and situation B is $\lim\limits_{\bullet\rightarrow0} \frac{\partial{h_{m(n-1)}(\bullet)}}{\partial{\bullet}} = -\infty$.

\noindent{\emph{\textbf{Proof for situation A ($\frac{\partial{h_{m(n-1)}(\bullet)}}{\partial{\bullet}}\in L_{\infty}$):}}}

By seeking a contradiction, it is to be proved that
\begin{equation}
|\xi|<1,\forall t\geq0
\label{eq_Toprove}
\end{equation}
so that $\textbf{Z}(t)\in\Psi_s$ according to the definition of $\xi$ in \eqref{eq_xi}.

Since the variables in the manifold constraint variable system \eqref{eq_dxi} are all continuous. According to the  intermediate value theorem for continuous functions, if the system violates the manifold constraint, i.e. exists $|\xi(t_v)|\geq1$, then there must exists a moment $t_c\in(0,t_v]$, which denotes the time instant when \eqref{eq_Toprove} is violated for the first time, that
\begin{equation}
\lim\limits_{t\rightarrow t_c^-}|\xi|=1
\label{eq_limtc}
\end{equation}
and
\begin{equation}
\lim\limits_{t\rightarrow t_c^-}\frac{\partial |\xi|}{\partial t}=\lim\limits_{|\xi|\rightarrow 1^-}\frac{\partial |\xi|}{\partial t}\geq0.
\label{eq_limtcdxi}
\end{equation}

It can guarantee that $\xi\in(-1,1)$ and all signals are bounded for $t\in[0,t_c)$ from \textbf{Lemma} \ref{lemma1}. The following discussions from \eqref{eq_dxi1} to \eqref{eq_limdxi1} are based on the time interval $[0,t_c)$.

According to \eqref{eq_xi}, the manifold constraint variable system can be expressed as:
\begin{equation}
\begin{aligned}
&\dot{\xi}=2\cdot\\
&\left\{
\begin{aligned}
&\frac{\dot{s}_{n-1}-\dot{x}_c-\dot{x}_L}{x_U-x_L}+\frac{(s_{n-1}-x_c-x_L)(\dot{x}_U-\dot{x}_L)}{(x_U-x_L)^2}\\
&\qquad ,\ \text{OMC}\\
&\frac{\dot{s}-\dot{y}_L}{y_U-y_L}+\frac{(s-y_L)(\dot{y}_U-\dot{y}_L)}{(y_U-y_L)^2}\\
&\qquad ,\ \text{LoMC}\\
&\frac{\dot{s}_{n-1}-h_v'\ddot{s}_{n-1}-\dot{x}_L}{x_U-x_L}+\frac{(s_{n-1}-h_v-x_L)(\dot{x}_U-\dot{x}_L)}{(x_U-x_L)^2}\\
&\qquad ,\ \text{LaMC}
\end{aligned}
\right.\\
&=2\cdot
\left\{
\begin{aligned}
&\frac{-\dot{x}_c}{x_U-x_L}+C_1,&\text{OMC}\\
&\frac{\dot{s}}{y_U-y_L}+C_2,&\text{LoMC}\\
&\frac{-h_v'\ddot{s}_{n-1}}{x_U-x_L}+C_3,&\text{LaMC}
\end{aligned}
\right.\\
\end{aligned}
\label{eq_dxi1}
\end{equation}
where $C_1=\frac{\dot{s}_{n-1}-\dot{x}_L}{x_U-x_L}+\frac{(s_{n-1}-x_c-x_L)(\dot{x}_U-\dot{x}_L)}{(x_U-x_L)^2}$, $C_2=\frac{-\dot{y}_L}{y_U-y_L}+\frac{(s-y_L)(\dot{y}_U-\dot{y}_L)}{(y_U-y_L)^2}$, and $C_3=\frac{\dot{s}_{n-1}-\dot{x}_L}{x_U-x_L}+\frac{(s_{n-1}-h_v-x_L)(\dot{x}_U-\dot{x}_L)}{(x_U-x_L)^2}$ are all bounded.

According to \eqref{eq_IS} and \eqref{eq_ISFM}, it can be obtained that
\begin{equation}
\dot{s}_{n-1}=s+h_{m(n-1)}(s_{n-1})=z_n-\sum_{i=1}^{n-2}h_{mi}^{(n-1-i)}(s_i)
\label{eq_dsn-1}
\end{equation}
and
\begin{equation}
\ddot{s}_{n-1}=\dot{z}_n-\sum_{i=1}^{n-2}h_{mi}^{(n-i)}(s_i)
\label{eq_ddsn-1}
\end{equation}

Simultaneously taking the derivative of both sides of equation \eqref{eq_xyc} yields
\begin{equation}
\begin{aligned}
\frac{\partial h_{m(n-1)}(x_c)}{\partial x_c}\dot{x}_c=&\frac{\dot{y}_Ux_U-y_U\dot{x}_U}{x_U^2}(x_c-s_{n-1})\\
&+\frac{y_U}{x_U}(\dot{x}_c-\dot{s}_{n-1})+\ddot{s}_{n-1}
\end{aligned}
\label{eq_dxc}
\end{equation}
Therefore, it can be obtained that
\begin{equation}
\begin{aligned}
\dot{x}_c&=\frac{\frac{\dot{y}_Ux_U-y_U\dot{x}_U}{x_U^2}(x_c-s_{n-1})-\frac{y_U}{x_U}\dot{s}_{n-1}+\ddot{s}_{n-1}}{\frac{\partial h_{m(n-1)}(x_c)}{\partial x_c}-\frac{y_U}{x_U}}\\
&=A_1\ddot{s}_{n-1}+B_1=A_1\dot{z}_n-A_1\sum_{i=1}^{n-2}h_{mi}^{(n-i)}(s_i)+B_1
\end{aligned}
\label{eq_dxceq}
\end{equation}
in which $A_1=\frac{1}{\frac{\partial h_{m(n-1)}(x_c)}{\partial x_c}-\frac{y_U}{x_U}}$ and $B_1=\frac{\frac{\dot{y}_Ux_U-y_U\dot{x}_U}{x_U^2}(x_c-s_{n-1})-\frac{y_U}{x_U}\dot{s}_{n-1}}{\frac{\partial h_{m(n-1)}(x_c)}{\partial x_c}-\frac{y_U}{x_U}}$ are both bounded and $-A_1$ has a positive lower bound since $\frac{\partial h_{m(n-1)}(x_c)}{\partial x_c}<0$.

By substituting \eqref{eq_ddsn-1}, \eqref{eq_dxceq}, and the controller \eqref{eq_uT}, the manifold constraint variable system \eqref{eq_dxi1} can be rewritten as:
\begin{equation}
\begin{aligned}
&\dot{\xi}=2\cdot\\
&\left\{
\begin{aligned}
&-\frac{A_1\dot{z}_n-A_1\sum_{i=1}^{n-2}h_{mi}^{(n-i)}(s_i)+B_1}{x_U-x_L}+C_1,&\text{OMC}\\
&\frac{\dot{z}_n-\sum_{i=1}^{n-1}h_{mi}^{(n-i)}(s_i)}{y_U-y_L}+C_2,&\text{LoMC}\\
&\frac{-h_v'\left(\dot{z}_n-\sum_{i=1}^{n-2}h_{mi}^{(n-i)}(s_i)\right)}{x_U-x_L}+C_3,&\text{LaMC}
\end{aligned}
\right.\\
&=A\dot{z}_n+D=-AGk_u\Gamma(\xi)+AF+D
\end{aligned}
\label{eq_dxi}
\end{equation}
where $h_v':=\frac{\partial h_v(\dot{s}_{n-1})}{\partial\dot{s}_{n-1}}$ and
\begin{equation}
\begin{aligned}
&A=2\cdot
\left\{
\begin{aligned}
&-\frac{A_1}{x_U-x_L},\ \text{OMC}\\
&\frac{1}{y_U-y_L},\ \text{LoMC}\\
&\frac{-h_v'}{x_U-x_L},\ \text{LaMC}
\end{aligned}
\right.
\end{aligned}
\label{eq_dxiA}
\end{equation}
\begin{equation}
\begin{aligned}
&D=2\cdot
\left\{
\begin{aligned}
&-\frac{-A_1\sum_{i=1}^{n-2}h_{mi}^{(n-i)}(s_i)+B_1}{x_U-x_L}+C_1,&\text{OMC}\\
&\frac{-\sum_{i=1}^{n-1}h_{mi}^{(n-i)}(s_i)}{y_U-y_L}+C_2,&\text{LoMC}\\
&\frac{h_v'\sum_{i=1}^{n-2}h_{mi}^{(n-i)}(s_i)}{x_U-x_L}+C_3,&\text{LaMC}
\end{aligned}
\right.
\end{aligned}
\label{eq_dxiD}
\end{equation}
Since $\frac{\partial{h_{m(n-1)}(\bullet)}}{\partial{\bullet}}<0$ and $\frac{\partial{h_{m(n-1)}(\bullet)}}{\partial{\bullet}}\in L_{\infty}$, the value $-h_v'$ is bounded with a positive lower bound. There exit positive constants $\underline{A}$, $\bar{A}$, $\bar{D}$, and $\bar{F}$ such that $0<\underline{A}\leq A\leq \bar{A}$, $|D|\leq\bar{D}$, and $|F|\leq\bar{F}$.

It holds that
\begin{equation}
\begin{aligned}
\frac{\partial |\xi|}{\partial t}\leq-\underline{A}\underline{G}k_u|\Gamma(\xi)|+\bar{A}\bar{F}+\bar{D}<0
\end{aligned}
\label{eq_dxileq}
\end{equation}
when $|\xi|>\Gamma^{-1}(\frac{\bar{A}\bar{F}+\bar{D}}{\underline{A}\underline{G}k_u})$,
and
\begin{equation}
\lim\limits_{|\xi|\rightarrow 1^-}\frac{\partial |\xi|}{\partial t}=\lim\limits_{|\Gamma(\xi)|\rightarrow \infty}\frac{\partial |\xi|}{\partial t}=-\infty.
\label{eq_limdxi1}
\end{equation}

This contradicts with \eqref{eq_limtcdxi}, therefore the assumption is not valid, $|\xi|<1$ for all time and the constraint is satisfied in situation A.

\noindent{\emph{\textbf{{Proof for situation B ($\lim\limits_{\bullet\rightarrow0} \frac{\partial{h_{m(n-1)}(\bullet)}}{\partial{\bullet}} = -\infty$):}}}
	
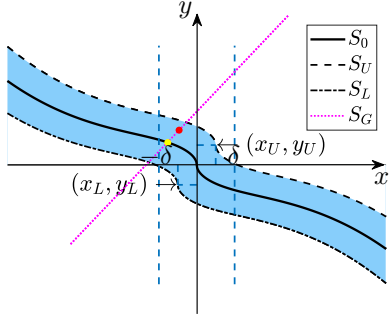
\begin{figure}[]
	\centering
	\begin{overpic}[width=0.8\columnwidth]{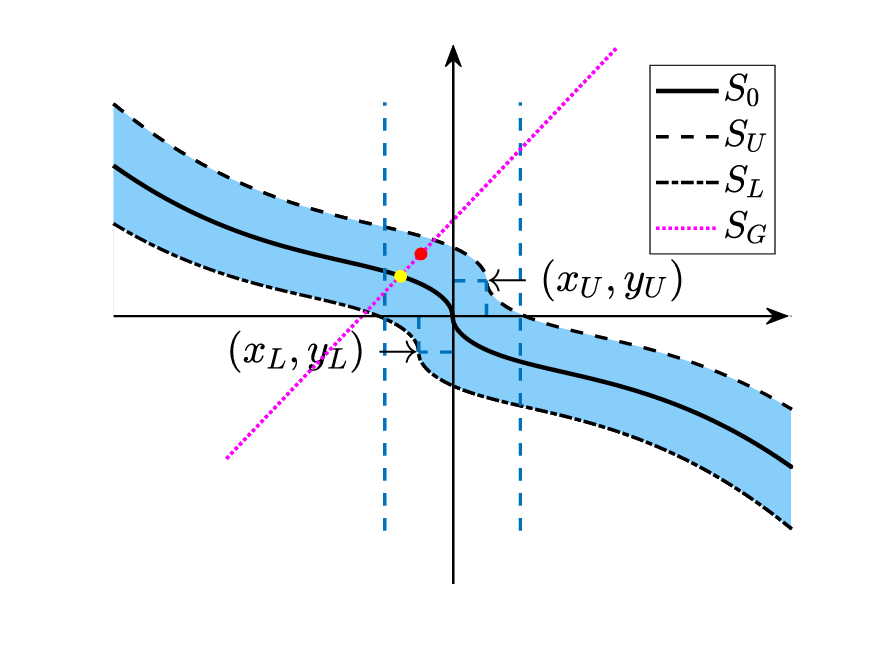}
		\put(88,35){$x$}
		\put(48,70){$y$}
		\put(40,39){$-\delta$}
		\put(58,39){$\delta$}
	\end{overpic}
	\caption{Manifold constraint diagram under situation B.}
	\label{fig_FiniteS2}
\end{figure}

For situation B, there exists a positive value $\delta$ that satisfies
\begin{equation}
h_{m(n-1)}(\delta-x_U)=-y_U.
\label{eq_Smdelta}
\end{equation}
$\frac{\partial{h_{m(n-1)}(s_{n-1})}}{\partial{s_{n-1}}}\in L_{\infty}$ for $s_{n-1} \notin \Omega_\delta\triangleq(-\delta,\delta)$. System fully actuated errors $\textbf{Z}$ maintains in $\Psi_s$ when $s_{n-1} \notin \Omega_\delta$ as the proof for situation A.

According to $\textbf{Z}(0)\in \Psi_s$ and the proof for situation A, there exists $t_\delta\geq0$ so that $\textbf{Z}\in \Psi_s$ and $s_{n-1} \in \Omega_\delta$. The red bullet in Fig. \ref{fig_FiniteS2} is the position $(s_{n-1}(t),\dot{s}_{n-1}(t))$ with $t\geq t_\delta$. The yellow bullet in Fig. \ref{fig_FiniteS2} is the position $(x_c(t),y_c(t))$.
Additionally, there exists a positive constant $\bar{\dot{s}}_{n-1}$ so that $ |\dot{s}_{n-1}|\leq \bar{\dot{s}}_{n-1}$ for $\textbf{Z}\in \Psi_s$ and $s_{n-1} \in \Omega_\delta$, and the horizontal distance $\delta_x$ between the red bullet and the boundary $S_U$ or $S_L$ under $s_{n-1} \in \Omega_\delta$ satisfies that $0<\delta_x(t_\delta)<2\delta$.
There exits a positive constant $\underline{t}_{mx}=2\delta(t_\delta)/\bar{\dot{s}}_{n-1}$ that the time for the red bullet to move to the boundary $t_{mx}$ satisfies that $t_{mx}>\underline{t}_{mx}$.
The time for the red bullet to move to the horizontal axis $\dot{s}_{n-1}=0$ is $t_{my}$.

According to \eqref{eq_ddsn-1},
\begin{equation}
\int_{\dot{s}_{n-1}(t_\delta)}^{0}\frac{d\dot{s}_{n-1}}{\ddot{s}_{n-1}}=\int_{t_\delta}^{t_\delta+t_{my}}dt
\label{eq_intdds}
\end{equation}
and
\begin{equation}
\begin{aligned}
t_{my}&=\int_{\dot{s}_{n-1}(t_\delta)}^{0}\frac{1}{\dot{z}_n-A_M}d\dot{s}_{n-1}\\
&=\int_{\dot{s}_{n-1}(t_\delta)}^{0}\frac{1}{-Gk_u\Gamma(\xi)+F-A_M}d\dot{s}_{n-1}
\end{aligned}
\label{eq_tmy}
\end{equation}
where $A_M=\sum_{i=1}^{n-2}h_{mi}^{(n-i)}(s_i)$ and it has an upper bound $\bar{A}_M$ according the properties of $h_{mi}$.
Therefore, in the worst-case scenario, the red bullet moves close to the boundary, $t_{my}<t_{mx}$ as long as $\Gamma^{-1}\left(\frac{\bar{\dot{s}}_{n-1}^2/(2\delta(t_\delta))+\bar{F}_d+\bar{A}_M}{\underline{G}k_u}\right)\leq|\xi|<1$, that is, the red bullet always can cross the horizontal axis $\dot{s}_{n-1}=0$ before touching the boundary. Then the red bullet will move in the opposite direction, making it impossible to cross the boundary.

Furthermore, since all signals are bounded for $t\geq0$, there exists a positive constant $\bar{\xi}$ such that $|\xi|\leq\bar{\xi}<1$, which implies that the controller has an upper bound $\widetilde{u}=k_u\Gamma(\bar{\xi})$.

\section{Proof of Theorem \ref{theoremSFIC}}
\label{Proof_theoremSFIC}

According to the designs in \eqref{eq_WxyUxyL}, when the distance between the current state \((s_{n-1}(t), ds_{n-1}(t)) \in \Psi_s\) and the boundaries $\textbf{S}_U,\textbf{S}_L$ greater than \(\rho_e\), the flexible constraint boundaries in \eqref{eq_WxyUxyL} equals the original constraint boundaries in \eqref{eq_xUxL} and \eqref{eq_yUyL}, and the flexible constraint variable $\widetilde{\xi}$ in \eqref{eq_Wxi} equals the original constraint variable $\xi$ in \eqref{eq_xi}, making the flexible controller in \eqref{eq_WuT} identical to the controller in \eqref{eq_uT}.

When the distance between the current state \((s_{n-1}(t), ds_{n-1}(t))\) and the manifold constraint boundaries $\textbf{S}_U$ and $\textbf{S}_L$ is less than \(\rho_e\), the flexible constraint boundaries $\widetilde{x}_U,\widetilde{x}_L,\widetilde{y}_U,\widetilde{y}_L$ in \eqref{eq_WxyUxyL} expand based on the original constraint boundary $x_U,x_L,y_U,y_L$ in \eqref{eq_xUxL} and \eqref{eq_yUyL}. After the current state exceeds the original manifold constraint boundaries $\textbf{S}_U$ and $\textbf{S}_L$, \(\mathcal{T}_s = 1\) in \eqref{eq_WxyUxyL} can be obtained, the flexible expanded constraint \eqref{eq_WxyUxyL} ensures that the distance between the current state and the flexible boundaries remains \(\rho_e\), so the absolute value of the flexible constraint variable $\widetilde{\xi}$ can be redefined as
\begin{equation}
|\widetilde{\xi}|=\frac{d_{sc}}{d_{sc}+\rho_e}<1.
\label{eq_|xi|}
\end{equation}
This guarantees that the flexible controller is well-defined and ensures the effective operation of the controller. If the system is not ISS and ISpS, as \(d_{sc}\) increases, \(|\widetilde{\xi}|\) continues to increase and approaches 1, while the flexible controller also continues to increase and tends to infinity.

Based on the above analysis, there must exists a sufficiently small positive constant \(\bar{\rho}_e\)
\begin{equation}
\bar{\rho}_e=\min\left\{d_U-\Gamma^{-1}\left(\frac{\max\{\bar{u},-\underline{u}\}}{k_u}\right)\right\},
\label{eq_brhoe}
\end{equation}
such that the controller will enter saturation before flexible expansion is activated for \(\rho_e\in(0,\bar{\rho}_e]\). In this case, if the actuator output capability is sufficient, the controller will not reach the input constraint \eqref{eq_sat}, thus achieving the same control effect as in Theorem \ref{theoremSF} (\romannumeral1) and (\romannumeral2). If the actuator output capability is insufficient, as the state approaches the constraint boundary, the system will first enter saturation, and then proceed with the flexible expansion of the constraint boundary as \eqref{eq_WxyUxyL}. In the reverse process, as the state \((s_{n-1}(t), ds_{n-1}(t))\) moves back from outside the original constraint set \(\Psi_s\) to within the set as \((s_{n-1}(t), ds_{n-1}(t))\in\Psi_s\), when the distance to the constraint boundary exceeds \(\rho_e\), the flexible constraint boundary will restore to the original constraint value, i.e., \(\widetilde{x}_U=-\widetilde{x}_L=x_U=-x_L\) and \(\widetilde{y}_U=-\widetilde{y}_L=y_U=-y_L\), since $\mathcal{T}_s=0$ in \eqref{eq_WxyUxyL}. This ensures that when the controller exits saturation, the flexible constraint boundary has been fully restored to the original constraint.

\section{Proof of Theorem \ref{theoremOF}}
\label{Proof_theoremOF}

Define the scaled estimation errors as:
\begin{equation}
\zeta_i=\frac{z_i-\hat{z}_i}{\mu^{n-i}},\quad i=1,\dots,n
\label{eq_SZ}
\end{equation}
Therefore, $\hat{\textbf{Z}}=\textbf{Z}-\mathcal{D}(\mu)\zeta$ with $\mathcal{D}(\mu)=diag(\mu^{n-1},\dots,\mu,1)$ and $\zeta=[\zeta_1,\dots,\zeta_n]^T$.

The derivative of $\zeta$ is
\begin{equation}
\dot{\zeta}=\frac{1}{\mu}(\mathcal{A}-H\mathcal{C})\zeta+\mathcal{B}(Gu+F)
\label{eq_dzeta}
\end{equation}
with $H=[a_1,a_2,\dots,a_n]^T$.

From Assumptions \ref{Assumption yd}-\ref{Assumption gf}, there exits positive constants $\bar{\xi}<1$, $\bar{\textbf{Z}}$, and $\tau_2$ that $(\xi,\textbf{Z})\in\Psi_{\xi Z},\forall t\in[0,\tau_2]$, where $\Psi_{\xi Z}$ is a set defined as $\Psi_{\xi Z}=\left\{(\xi,\textbf{Z})| |\xi|\leq\bar{\xi},\|\textbf{Z}\|\leq\bar{\textbf{Z}}\right\}$. Hence, it can be found positive constant $\bar{\Delta}$ that $Gu+F\leq\bar{\Delta}, \forall t\in[0,\tau_2]$.

Define Lyapunov function $V_{\zeta}=\zeta^T\mathcal{P}\zeta$ where positive definite symmetric matrix $\mathcal{P}\in\mathcal{R}^{n\times n}$ is the solution of $\mathcal{P}(\mathcal{A}-H\mathcal{C})+(\mathcal{A}-H\mathcal{C})^T\mathcal{P}=-\textbf{I}_{n}$.
Differentiating $V_{\zeta}$ with respect to time yields
\begin{equation}
\begin{aligned}
\dot{V}_\zeta
=&-\frac{1}{\mu}\zeta^T\zeta+2\zeta^T\mathcal{P}\mathcal{B}(Gu+F)\\
\leq&-\frac{1}{\mu}\|\zeta\|^2+2\|\zeta\|\|\mathcal{P}\|\bar{\Delta}\\
\leq&-\frac{1}{\mu}\|\zeta\|^2+\frac{1}{2\mu}\|\zeta\|^2+2\mu\|\mathcal{P}\|^2\bar{\Delta}^2\\
\leq&-\frac{1}{2\mu\|\mathcal{P}\|}V_{\zeta}+2\mu\|\mathcal{P}\|^2\bar{\Delta}^2.
\end{aligned}
\label{eq_dVzeta}
\end{equation}
Since the convergence rate of $V_{\zeta}$ will faster and faster as well as the convergence error will smaller and smaller when $\mu$ tends to $0$, there must be a small positive constant $\mu_1$ and $\tau_1<\tau_2$ to make $\|\zeta\|\leq\bar{\zeta}$ when $\mu\in(0,\mu_1)$ and $t\in(\tau_1,\tau_2]$ for any positive constant $\bar{\zeta}$.

Considering that the controller structure is Lipschitz, there must exist positive constants $\mathcal{L}$ and $\mu_2<1$ for any $\mu\in(0,\mu_2)$ such that
\begin{equation}
\|v(\hat{\textbf{Z}})-v(\textbf{Z})\|=\|v(\textbf{Z}-\mathcal{D}(\mu)\zeta)-v(\textbf{Z})\|\leq\mathcal{L}\|\zeta\|.
\label{eq_uOFD}
\end{equation}

Therefore there exists a positive constant $\mu_3$ such that for $\mu\in(0,\mu_3)$, $\|\zeta\|$, $\|\hat{\textbf{Z}}-\textbf{Z}\|$, and $\mathcal{L}$ are all sufficiently small, hence \textbf{Theorem} \ref{theoremSFIC} can be used to obtain the same conclusion since the error of the $v$ is bounded. Thus $\textbf{Z}(t)\in\Psi_s$ holds for $t\geq0$ when $0<\mu<\min\left(\mu_1,\mu_2,\mu_3\right)=\bar{\mu}$.


\section*{References}

\bibliographystyle{IEEEtran}
\bibliography{FAMC}

\end{document}